\documentclass[journal,draftclsnofoot]{IEEEtran}

\onecolumn
\usepackage[utf8]{inputenc} 
\usepackage[T1]{fontenc}
\usepackage{url}
\usepackage{ifthen}
\usepackage{cite}
\usepackage[cmex10]{amsmath}

\usepackage{nccmath}
\usepackage{microtype}
\usepackage{amssymb}
\usepackage{mathtools}
\usepackage{amsthm}
\usepackage{comment}
\usepackage[export]{adjustbox}

\usepackage[hidelinks]{hyperref}

\usepackage{lipsum}

\usepackage{pgfplots}
\pgfplotsset{compat=1.15}
\usepackage{graphicx}
\usepackage{pstool}
\usepackage{tikz}
\usetikzlibrary{positioning}
\usetikzlibrary{arrows,shapes,math}
\usepackage{epic}
\usepackage{epstopdf}
\pgfdeclarelayer{bg}    % declare background layer
\pgfsetlayers{bg,main}

\newcommand*{\PROOFS}{}% %comment this line to remove appendix and references to it. 

\newtheorem{thm}{Theorem}

\newtheorem{cor}[thm]{Corollary}
\newtheorem{lem}[thm]{Lemma}
\newtheorem{prop}[thm]{Proposition}

\theoremstyle{definition}

\newtheorem{assumption}{Assumption}

\newtheorem{definition}{Definition}

\tikzstyle{sum}=[circle, fill=blue!10, draw=black,line width=1pt,minimum size = 0.5cm, thick ]
\tikzstyle{ssum}=[circle, fill=blue!10,draw=black,line width=1pt,minimum size = 0.1cm]
\tikzstyle{ssum_red}=[circle, fill=red!10,draw=black,line width=1pt,minimum size = 0.1cm]
\tikzstyle{int1}=[draw, fill=blue!10, minimum height = 0.5cm, minimum width=0.5cm,thick ]
\tikzstyle{int1_red}=[draw, fill=red!10, minimum height = 0.5cm, minimum width=0.5cm,thick ]
\tikzstyle{int}=[draw, fill=blue!10, minimum height = 0.8cm, minimum width=0.7cm,thick ]

\definecolor{mygrey}{RGB}{229,229,229}
\definecolor{mygrey2}{RGB}{127,127,127}
\definecolor{mygrey3}{RGB}{240,240,240}
\pgfplotsset{
 	axis background/.style={fill=mygrey},
	tick style=mygrey2,
	tick label style=mygrey2,
	grid=both,
	xtick pos=left,
	ytick pos=left,
	tick style={
		major grid style={style=white,line width=1pt},minor grid style=mygrey3,
		tick align=outside,
	},
	minor tick num=1,
}

\newcommand*{\QEDA}{\hfill\ensuremath{\square}}

\newcommand*{\Var}{\mathrm{Var}}

\newcommand*{\Mcal}{\mathcal{M}}
\newcommand*{\Ncal}{\mathcal{N}}
\newcommand*{\Ccal}{\mathcal{C}}
\newcommand*{\Scal}{\mathcal{S}}

\newcommand{\cC}{\mathcal{C}}
\newcommand{\cP}{\mathcal{P}}

\newcommand*{\sph}{\mathsf{sp}}

\newcommand*{\Prm}{P}
\newcommand*{\one}{\mathbf{1}}
\newcommand{\Dsp}{D^{\sph}}

\newcommand{\bEx}{\ensuremath{\mathbb{E}}}
\newcommand{\dd}{\mathrm{d}}

\DeclareMathOperator{\var}{\mathsf{Var}}

\newcommand{\eps}{\epsilon}
\newcommand{\ex}[1]{\ensuremath{\mathbb{E}\left[ #1\right]}}
\newcommand{\pr}[1]{\ensuremath{\mathbb{P}\left[ #1\right]}}

\newcommand{\reals}{\mathbb{R}}
\newcommand{\normal}{\mathcal{N}}
\newcommand{\cM}{\mathcal{M}}

\newcommand{\cS}{\mathcal{S}}

\DeclareMathOperator{\Lip}{Lip}
\DeclareMathOperator{\cov}{\mathsf{Cov}}

\usepackage{enumerate}
\let\originalleft\left
\let\originalright\right
\renewcommand{\left}{\mathopen{}\mathclose\bgroup\originalleft}
\renewcommand{\right}{\aftergroup\egroup\originalright}
\newif\iflongpaper
\longpapertrue

\author{%
  \IEEEauthorblockN{Alon Kipnis\IEEEauthorrefmark{1}}
  and
  \IEEEauthorblockN{Galen Reeves\IEEEauthorrefmark{2}}
  
  \IEEEauthorblockA{\IEEEauthorrefmark{1}Department of Statistics,
  Stanford University, Stanford, CA 94305 USA}
  
  \IEEEauthorblockA{\IEEEauthorrefmark{2}Department of Electrical and Computer Engineering and 
Department of Statistical Science, Duke University, 
Durham, NC 27708 USA}

\thanks{This paper was presented in part at the IEEE International Symposium on Information Theory (ISIT), 2019 \cite{KipnisReevesISIT2019}. 
}
}

\title{%\LARGE \bf
Gaussian Approximation of Quantization Error for
Estimation from Compressed Data
}

\begin{document}

\maketitle

\begin{abstract}
We consider the distributional connection between the lossy compressed representation of a high-dimensional signal $X$ using a random spherical code and the observation of $X$ under an additive white Gaussian noise (AWGN). We show that the Wasserstein distance between a bitrate-$R$ compressed version of $X$ and its observation under an AWGN-channel of signal-to-noise ratio $2^{2R}-1$ is bounded in the problem dimension. We utilize this fact to connect the risk of an estimator based on the compressed version of $X$ to the risk attained by the same estimator when fed the AWGN-corrupted version of $X$. We demonstrate the usefulness of this connection by deriving various novel results for inference problems under compression constraints, including  minimax estimation, sparse regression, compressed sensing, and universality of linear estimation in remote source coding.
\end{abstract}


\section{Introduction}
%Due to the disproportionate size of modern datasets compared to available computing and communication resources, many inference techniques are applied to a compressed representation of the data rather than the data itself. 
%
%In this paper we consider the task of estimating the latent parameter $\theta$ from a lossy compressed version of the data $X$, as illustrated in Figure~\ref{fig:intro}. 
%
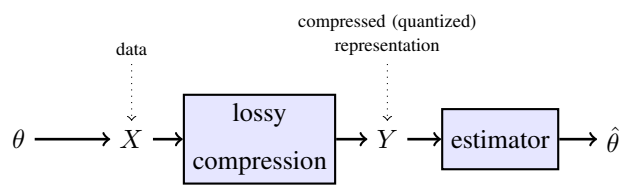
\begin{figure}[b]
\centering
\begin{tikzpicture}[scale =1]
\node (x) at (0,0) {$X$};
%\node[int1_red, left of = x, node distance = 1.2cm] (pxy) {$\Prm_{X|\theta}$};
\node[left of = x, node distance = 1.5cm] (theta) {$\theta$};
\draw[->, line width = 1pt] (theta)--(x);
%\draw[->, line width = 1pt] (pxy)--(x);

\node[int1, align = center, scale = 1, right of = x, node distance = 1.7cm] (comp) {lossy \\ compression};

\node[right of = comp, node distance = 1.7 cm] (y) {$Y$};

\node[int, right of = y, node distance = 1.5cm] (est) {estimator};

\draw[->,dotted] (x)+(0,1) node[above] {\scriptsize data} -- (x);

\draw[->,dotted] (y)+(0,1) node[above,align=center,execute at begin node=\setlength{\baselineskip}{2ex}] {\scriptsize compressed (quantized)\\ \scriptsize representation} -- (y);

\node[right of = est, node distance = 1.5 cm] (thh) {$\hat{\theta}$};

\draw[->, line width = 1pt] (x) -- (comp);
\draw[->, line width = 1pt] (comp) -- (y);
\draw[->, line width = 1pt] (y) -- (est);
\draw[->, line width = 1pt] (est) -- (thh);

\end{tikzpicture}
\caption{
Inference about the latent signal $\theta$ is based on degraded observations $Y$ of the data $X$.
\label{fig:intro}}
\end{figure}
\begin{figure}[t]
\centering
\begin{tikzpicture}[scale =1]
\node (x) at (0,0) {$X$};
%\node[int1_red, left of = x, node distance = 1.2cm] (pxy) {$\Prm_{X|\theta}$};
\node[left of = x, node distance = 1.5cm] (theta) {$\theta$};
\draw[->, line width = 1pt] (theta)--(x);
%\draw[->, line width = 1pt] (pxy)--(x);

\node[int1, align = center, scale = 1] at (1,1) (enc) {Enc};
\node[int1, right of = enc, node distance = 2.8cm, align = center, scale = 1](dec) {Dec};

\begin{pgfonlayer}{bg}    % select the background layer
        \fill[fill = gray!20] (enc)+(-0.75,-0.5) rectangle +(3.5,0.7);
        \node[above right = 0.35cm and -0.3cm of enc] {\small lossy compression};
\end{pgfonlayer}

\node[right of = dec, node distance = 1.2 cm] (y) {$Y$}; %{$\hat{\theta}$};

\node[circle, draw, align = center, scale = 1] at (2.3,-1)(plus) {+};
\node[above of = plus, node distance = 1cm, scale = 1] (noise) {$\normal(0, \sigma^2 I_n)$};
\node[below of = y, node distance = 2cm] (z) {$Z$};

\draw[->, line width = 1pt] (x) |- (enc);
\draw[->, line width = 1pt] (enc) -- (dec)node[midway, above]{$ \{1,\dots, M\}$};
\draw[->, line width = 1pt] (dec) -- (y);
\draw[->, line width = 1pt] (x) |-(plus);
\draw[->, line width = 1pt] (plus) -- (z);
\draw[->, line width = 1pt] (noise) -- (plus);

\end{tikzpicture}
\caption{
The effect of a bitrate constraint is compared to the effect of additive Gaussian noise by studying the Wasserstein distance between $\Prm_{Y}$ and $\Prm_{Z}$. Under random spherical encoding, we show that this distance is bounded in the problem dimension $n$, hence estimating $\theta$ from $Y$ is equivalent to estimating it from $Z$. 
\label{fig:contribution}}
\end{figure}
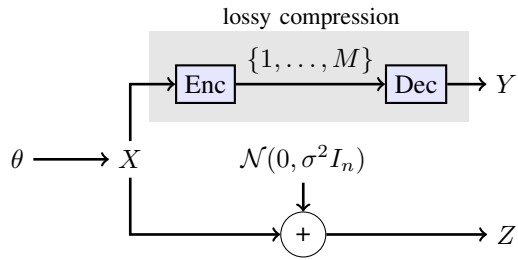

Due to the disproportionate size of modern datasets compared to available computing and communication resources, many inference techniques are applied to a compressed representation of the data rather than the data itself (Figure~\ref{fig:intro}). In the attempt to develop and analyze inference techniques based on a degraded version of the data, it is conceptually appealing to model inaccuracies resulting from lossy compression as additive noise. Indeed, there exists a rich literature devoted to the characterization of this ``noise'', i.e., the difference between the original data and its compressed representation \cite{gray1998quantization}. Nevertheless, because of the difficulty of analyzing non-linear compression operations, this characterization is generally limited to the high-bit compression regime and other restrictions on the distribution of the data
\cite{BLTJ1340, 1424312, 370112, lee1996asymptotic, kontoyiannis2006mismatched}. \par %As 
In this paper, we establish a strong and relatively simple characterization of the distribution of quantization error corresponding to a random spherical code. Specifically, we show that, in the sense of the Wasserstein distance, this error can be approximated by additive white Gaussian noise (AWGN) whose variance $\sigma^2$ is inversely proportional to $2^{2R}-1$ where $R$ is the bitrate of the code (Figure~\ref{fig:contribution}). This approximation implies that the expected error of an estimator applied to the compressed representation of the data is asymptotically equivalent to the  expected error of the same estimator applied to a Gaussian noise-corrupted version of the data. 
The benefit from such approximation is twofold: (1) inference techniques from  observations corrupted by Gaussian noise can now be applied directly to the compressed representation; and (2) it provides a mechanism to characterize the performance of inference using such techniques. 

\subsection{Overview of Main Contributions}
The equivalence illustrated in Figure~\ref{fig:contribution} allows us to derive various novel results for two closely related inference settings, both of which are performed on a lossy compressed representation of the observed data $X = (X_1, \dots, X_n)$.
\begin{itemize}
\item\emph{Parameter Estimation:} (Section~\ref{sec:parameter_estimation}) The data are drawn according to a distribution indexed by an unknown $d$-dimensional parameter vector $\theta$ and the goal is to estimate the parameter vector under the squared error loss. In the high-dimensional setting, the number of parameters $d$ is possibly much larger than the number of observations $n$. This problem is also related to learning distributions under communication constraints \cite{tsitsiklis1988decentralized,zhang1988estimation, HanAmari1998, steinhardt2015minimax,zhu2017quantized,KipnisDuchi,han2018distributed, dagan2018detecting,  szabo2018adaptive, barnes2019learning}. 
\item \emph{Indirect Source Coding:} (Section~\ref{sec:source_coding}) The data are distributed jointly with an unknown random (source) vector $U = (U_1,\ldots,U_n)$ and the goal is to reconstruct this vector from the compressed representation of $X$ \cite{DobrushinTsybakov,berger1971rate,1054469, 1056251}.
\end{itemize}
At a high level, the main difference between these inference tasks is that the source coding problem assumes a joint distribution over the data and the quantities of interest. Beyond these settings, one may may also consider minimax and universal source coding formulation \cite{dembo2003minimax, weissman2004universally} as well as hypotheses testing \cite{HanAmari1998,tsitsiklis1988decentralized}. \par
In the parameter estimation setting, we consider the minimax mean-squared error (MSE) 
\begin{align*}
\cM_{n}^*  \coloneqq
\inf_{\phi,\psi}  \sup_{\theta \in \Theta_n} \frac{1}{d_n}  \bEx \left[  \| \theta  - \psi(\phi(X))\|^2 \right] ,
\end{align*}
where the infimum is over all encoding functions $\phi \colon \reals^n \to \{1, \dots, M\}$ and decoding functions $\psi\colon \{1, \dots, M\} \to \reals^d$ with $M = \lceil 2^{nR} \rceil$.
Zhu and Lafferty~\cite{zhu2014quantized} provided an asymptotic expression for $\Mcal^*_n$ in the special case of the Gaussian location model $X \sim \Ncal(\theta, \eps^2 I_n)$ where the parameter space $\Theta_n$ is an $n$-dimensional ball. Under a similar setting, our main results yield a non-asymptotic upper bound to $\Mcal^*_n$. Furthermore, under the additional assumption that $\theta$ is $k$-sparse, our main results implies that $\Mcal^*_n$ is upper bounded by a univariate function describing the minimax risk of soft-thresholding in sparse estimation \cite{Donoho1994, johnstone2011gaussian}. 
%In particular, the limit $\epsilon^2 \to 0$ of the resulting expression provides an achievable MSE for the problem of encoding a sparse signal considered in \cite{weidmann2012rate}. 
%
Finally, we consider the case where the data $X$ and the parameter $\theta$ are described by the model $X \sim \Ncal(A \theta , \epsilon^2 I)$, where $A \in \reals ^{d \times n}$ is a random matrix with i.i.d. Gaussian entries.
This setting with $\theta$ sparse and $d_n$  much larger than $n$ was studied in the context of the compressed sensing signal acquisition frameworks \cite{eldar2012compressed}. By applying our main results to estimation with the approximate message passing (AMP) algorithm \cite{donoho2009message},
we provide an exact asymptotic characterization of the MSE in recovering $\theta$ from a lossy compressed version of $X$ obtained using bitrate-$R$ random spherical coding. Versions of this compression and estimation problem for other type of lossy compression codes and estimators were considered in \cite{goyal2008compressive,baraniuk2017exponential,kipnis2017fundamental,8356140}. \par
The indirect source coding setting (other names are \emph{remote} or \emph{noisy} source coding and \emph{rate-constrained denoising}) corresponds to the case where $\{(U_i,X_i)\}_{i=1}^\infty$  %\nr{GR: I am not sure what the index set is here. Should we write $\{(U_i, X_i)\}_{i \in \mathbb{N}}$. Or at least we need to define what an information source is.}
is an ergodic process 
%information source 
with a finite second moment. A bitrate-$R$ spherical code is applied to $X = \{X_i\}_{i=1}^n$ while the goal is to estimate $\{U_n\}_{i=1}^n$ from the output $Y$ of this code \cite{DobrushinTsybakov} \cite{1054469} \cite[Ch 3.5]{berger1971rate}\cite{1056251}. 
For data normalized as $\ex{\|X\|^2} = n$, our main results imply that the MSE attained by any sequence of Lipschtiz estimators converges to the MSE attained by these estimators when applied to $\{Z_i\}_{i=1}^n$, where 
\begin{align}
    \label{eq:Z_channel_intro}
 Z_i = X_i + \sigma W_i, \quad \sigma^2 =  \frac{1}{2^{2R}-1},
 \end{align}
 Specialized to the case $X_i \mid U_i \sim \Ncal(U_i,\eps^2)$, our result implies an interesting universality property of spherical coding followed by  linear estimation: the resulting MSE equals the minimal MSE, over all encoding and estimation schemes, when a Gaussian source of the same second moment is estimated from a bitrate-$R$ encoded version of its observations under AWGN. This fact can be seen as a direct extension of the saddle point property of the Gaussian distribution in the standard (direct) source coding setting discussed in \cite{sakrison1968geometric, wyner:1968, lapidoth1997role, zhou2017refined}. 
 
\subsection{Background and Related Works}
Spherical codes have multiple theoretical and practical uses in numerous fields \cite{delsarte1991spherical}. In the context of information theory, Sakrison \cite{sakrison1968geometric} and Wyner \cite{wyner:1968} provided a geometric understanding of random spherical coding in a Gaussian setting; our main result extends their insights. Specifically, consider the representation of an $n$-dimensional standard Gaussian vector $X$ using $M = \lceil 2^{nR} \rceil$ codewords uniformly distributed over the sphere of radius $r = \sqrt{n (1-2^{-2R})}$. The left side of Fig.~\ref{fig:interpretation}, adapted from \cite{sakrison1968geometric}, shows a conceptual relation between $X$ and its nearest codeword $\hat{X}$: As $n$ increases, the angle $\alpha^*$ between the two concentrates so that $\sin(\alpha^*)$ converges to $2^{-R}$ in probability, 
%; the exact limiting distribution can be deduced from \cite{cai2012phase}. 
hence the quantized representation of $X$ and the error $X-\hat{X}$ become orthogonal. Consequently, the MSE between $X$ and its quantized representation, averaged over all random codebooks, converges to the Gaussian distortion-rate function $2^{-2R}$. %
In fact, as noted in \cite{lapidoth1997role}, this Gaussian coding scheme\footnote{We denote this scheme as \emph{Gaussian} since $r$ is chosen according to the distribution attaining the Gaussian DRF \cite[Ch. 10.5]{ThomasCover}. } achieves the Gaussian DRF when $X$ is generated by any ergodic information source of unit variance, implying that the second moments of $X- \hat{X}$ are independent of the distribution of $X$ as the problem dimension $n$ goes to infinity. \par
In this paper, we show that a much stronger statement holds for a properly scaled version of the quantized representation ($Y$ in Fig.~\ref{fig:interpretation}): in the limit of high dimension, the distribution of $Y-X$ is independent of the distribution of $X$ and is approximately Gaussian. 
This property of $Y-X$ suggests that the underlying quantities of interest (e.g.,  the parameter vector $\theta$ or the sequence $\{U_n\}_{n=1}^\infty$) can now be estimated as if $X$ is observed under additive Gaussian noise. This paper formalizes this intuition by showing that estimators from the Gaussian-noise corrupted version of $X$ ($Z$ in Fig.~\ref{fig:interpretation}) attain similar performances if applied to the scaled representation $Y$. 

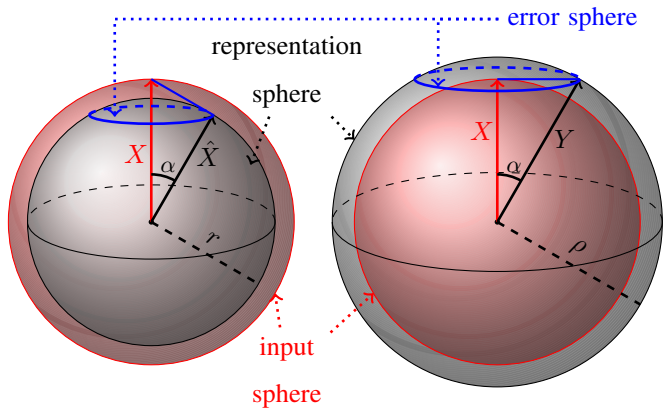
\begin{figure}
    \centering
    \begin{tikzpicture}
    \def\x0{-4.6}
    \def\s{1.9}
    \def\r{\s * 0.866}
  \shade[ball color = red!40, opacity = 0.4] (\x0,0) circle (\s);
  \shade[ball color = gray!40, opacity = 0.4] (\x0,0) circle (\r);
  \draw[color = red] (\x0,0) circle (\s);
  \draw (\x0,0) circle (\r);
  \draw (\x0-\r,0) arc (180:360:{\r} and 0.3*\r);
  \draw[dashed] (\x0 + \r,0) arc (0:180:{\r} and 0.3*\r);
  \fill[fill=black] (\x0,0) circle (1pt);
  %\draw[dashed, color = red, line width = 1pt] (\x0,0) -- node[above]{$\sqrt{n}$} (\x0-\s,0);
  \draw[dashed, line width = 1pt] (\x0,0) -- node[above, rotate = -30, xshift = -0.0cm]{$r$} ({\x0 + \r * 0.866},-{\r*0.5});
  
  \draw[color = red,line width = 1pt,->] (\x0,0) -- node[above, xshift = -0.2cm, yshift = -0.3cm]{$X$} (\x0,\s);
  \def\cosx{0.866}
  \def\sinx{0.5}
  \draw[color = black,line width = 1pt,->] (\x0,0) -- node[above, xshift = 0.35cm, yshift = -0.1cm]{$\hat{X}$} ({\x0 +\r*\sinx},{\r*\cosx});
  \draw[color = blue,thick] (\x0,\s) --  ({\x0 + \r*\sinx},{\r*\cosx});
  
  \draw[color = blue, line width = 1pt] (\x0+\r*\sinx,\r*\cosx) arc (0:-180:{\r*\sinx} and \r/14);
  \draw[color = blue, dashed, line width = 1pt] (\x0-\r*\sinx,\r*\cosx) arc (180:0:{\r*\sinx} and \r/14);
  
  %\draw[color = blue, dashed, line width = 1pt] (-\r*\sinx,\r*\cosx) arc (180:0:{\r*\sinx} and \r/10);
  
  \node at (\x0,-\s*9/8) {};
  \draw[color = black, line width = 1pt] (\x0,\s/3) node[right, yshift = 0.1cm, scale =.9] { $\alpha$} arc (90:60:{\s/3});
  
%%%%%%%%%%%  right ball %%%%%%%%%%%
    \def\s{1.9}
    \def\r{\s * 1/0.866}
  \shade[ball color = gray!40, opacity = 0.4] (0,0) circle (\r);
  \shade[ball color = red!40, opacity = 0.4] (0,0) circle (\s);
  \draw[color = red] (0,0) circle (\s);
  \draw (0,0) circle (\r);
  \draw (-\r,0) arc (180:360:{\r} and 0.3*\r);
  \draw[dashed] (\r,0) arc (0:180:{\r} and 0.3*\r);
  \fill[fill=black] (0,0) circle (1pt);
 % \draw[dashed, color = red, line width = 1pt] (0,0) -- node[above]{$\sqrt{n}$} (-\s,0);
 
 \draw[dashed, line width = 1pt] (0,0) -- node[above, rotate = -30, xshift = -0.0cm]{$\rho$} ({\r * 0.866},-{\r*0.5});
  
  \draw[color = red,line width = 1pt,->] (0,0) -- node[above, xshift = -0.2cm]{$X$} (0,\s);
  \def\cosx{0.866}
  \def\sinx{0.5}
  \draw[color = black, line width = 1pt,->] (0,0) -- node[above, xshift = 0.35cm, yshift = -0.1cm]{$Y$} ({\r*\sinx},{\r*\cosx});
  \draw[color = blue,thick] (0,\s) --  ({\r*\sinx},{\r*\cosx});
  
  \draw[color = blue, line width = 1pt] (\r*\sinx,\r*\cosx) arc (0:-180:{\r*\sinx} and \r/14);
  \draw[color = blue, dashed, line width = 1pt] (-\r*\sinx,\r*\cosx) arc (180:0:{\r*\sinx} and \r/14);
  %\draw[color = blue, dashed, line width = 1pt] (-\r*\sinx,\r*\cosx) arc (120:60:{\r} and \r*0.7);
  
  \node[align = center] at (-2.8,2) (rep) {representation \\ sphere};
  \node[align = center, color = red] at (-2.8,-2) (in) {input \\ sphere};
  
  \draw[->, dotted,  line width = 1pt] (rep) -- (-\r*\cosx,\r*\sinx);
  \draw[->, dotted, line width = 1pt] (rep) -- (\x0+\r*\cosx*25/36,\r*\sinx*25/36);
  
  \draw[->, color= red, dotted, line width = 1pt] (in) -- (-\s*\cosx,-\s*\sinx);
  \draw[->, color= red, dotted, line width = 1pt] (in) -- (\x0+\s*\cosx,-\s*\sinx);
  
  \draw[color = black, line width = 1pt] (0,\s/3) node[right, yshift = 0.1cm, scale =.9] {$\alpha$} arc (90:60:{\s/3});
  
  \node[align = center, color = blue] at (1,2.7) (err) {error  sphere};
 
  \draw[color = blue, ->, dotted,line width = 1pt] (err.west) -| +(\x0-0.5,-1.3);
  \draw[color = blue, ->, dotted,line width = 1pt] (err.west) -| +(-0.8,-0.9);
  \end{tikzpicture}
   \caption{Conceptual 3D illustration of random spherical coding in high-dimension. The norm of the input vector $X$ concentrates around the input sphere. 
   Left Sphere: Geometric interpretation of standard source coding from \cite{sakrison1968geometric,wyner:1968}. The representation sphere is chosen such that the error vector $\hat{X} - X$ is orthogonal to the reconstruction vector $\hat{X}$. Right Sphere: Geometric description of the quantization error considered in this paper. The representation sphere is chosen such that  $Y-X$ is orthogonal to $X$.
    }
    \label{fig:interpretation}
\end{figure}

In general, the radius of the codebook under which the distribution of $Y$ and $Z$ are close depends on the magnitude of $X$. 
%value around which $\|X\|^2/n$ concentrates.\nr{GR: Where do we assume it concentrates?}
This magnitude is only needed at the decoder, while the encoder can represent its input $X$ using codewords living, say, on the unit sphere. In particular, such an encoder is agnostic to the relationship between $X$ and $\theta$. This situation is in contrast to optimal quantization schemes in indirect source coding \cite{DobrushinTsybakov,berger1971rate} and in problems involving estimation from compressed data \cite{han1987hypothesis, zhang1988estimation, HanAmari1998, duchi2014optimality, han2018distributed, KipnisDuchi}, where the specification of the model $\theta \to X$ is crucial for designing the compression and estimation schemes. As a result, the random spherical coding scheme we rely on is sub-optimal in general, although it can be applied in situations where the model $\theta \to X$ is unknown at the compressor. Coding schemes with similar properties were studied in the context of indirect source coding under the name \emph{compress-and-estimate} in \cite{kipnis2021rate, KipnisWiener2019}.\par 
The equivalence between quantization noise and AWGN we provide in this paper is given in terms of the Wasserstein distance between the distributions of these vectors. We refer to \cite{rachev1998mass, ambrosio2003optimal, villani2008optimal} for properties, applications, and the long list of alternative names of the Wasserstein distance. 
In the context of information theory, the Wasserstein distance has been used to establish consistency of some quantization procedures
\cite{pollard1982quantization, linder2002lagrangian} and to define a class of channels over which communication is possible without assuming synchronization \cite{1056045, gray1980block}. One of the core results of this paper is a novel coupling of the distributions of $Y$ and $Z$ given $X$, leading to a bound on the Wasserstein distance between them. This bound, in combination with the fact that the $L_p$ risk of a Lipschitz estimator is continuous with respect to the Wasserstein distance, implies that the risk of such an estimator, when used at the output of a random spherical code, converges to the risk when used at the output of a Gaussian channel. 

% {\color{blue} 
% Our results about the equivalence between quantization noise and AWGN is similar to those obtained by Zamir and Feder \cite{ZamirFeder1996} for the quantization error of a random dithered lattice quantizer. Specifically, \cite{ZamirFeder1996} implies that the gap between this error and the standard Gaussian distribution behaves like $O(\log(n))$ in the Kullback-Leibler (KL) divergence sense. Using a transportation-cost inequality \cite{8187332}, this KL divergence bound implies a bound of the form $\sigma \cdot O(\sqrt{\log(n)})$ on the quadratic Wasserstein distance. 
% %
% In contrast, our results regarding quantization with spherical codes imply that the quadratic Wasserstein distance is bounded by $\sigma$ times a constant that is independent of $n$. 
% %
% Namely, both bounds have the proportionality parameter $\sigma$ associated with the quantization error, although the bound following from \cite{ZamirFeder1996} is unbounded in the problem dimension. 

% }

%{\color{blue}
Our work is similar in spirit to the work of Zamir and Feder \cite[Section~III]{ZamirFeder1996}, who provide a Gaussian approximation for the quantization error of a random dithered lattice quantizer. In the setting of their paper, the quantization error is independent of the input and distributed uniformly over the basic cell of the lattice. They show that as the dimension $n$ increases,  there exists a sequence of lattice quantizers such that the relative entropy between the distribution of the quantization noise and the isotropic Gaussian distribution with matched power is bounded from above by $c \cdot \log (n)$ where $c$ is a positive constant.  Normalizing by $n$, they conclude that the relative entropy per dimension converges to zero in the large-$n$ limit. To interpret their results in the setting of this paper, we can use the Gaussian transportation inequality \cite{RaginskySason2014}, which leads to an upper bound on the $2$-Wasserstein distance that is order $\sigma \sqrt{\log n}$ where $\sigma^2$ is the variance of the additive noise. By contrast, for quantization using a random spherical code, our results provide an upper bound on $p$-Wasserstein distance that is order $\sigma p$. Namely, both bounds are proportional to $\sigma$ but the bound following from \cite{ZamirFeder1996} is unbounded in the problem dimension.

Another related setting is the problem of channel simulation, the goal of which is to design a random code that induces a particular target distribution between the data and the compressed representation \cite{harsha2007communication, cuff2013distributed, LiElGamal2018}. \par
The rest of this paper is organized as follows. In Section~\ref{sec:main} we provide our main results on the distributional connection between spherical coding and AWGN. In Sections \ref{sec:parameter_estimation} and \ref{sec:source_coding} we apply these results to parameter estimation and source coding, respectively.  Section~\ref{sec:main_proof} provides the proofs of the main results. Concluding remark are provided in Section~\ref{sec:conclusions}. 

\section{Main Results \label{sec:main}}
The main result of this paper is a comparison between the quantization error under random spherical coding and independent Gaussian noise. 

\begin{definition}[Random Spherical Code] 
An $(n,M)$ random spherical code is a collection of $M$ codewords $\cC = \{C(1), \dots, C(M)\}$ drawn independently from the uniform distribution on the unit sphere in $\reals^n$. The encoder maps an input vector $x \in \reals^n \setminus \{0\}$ to the index $i^* \in \{1,\dots, M\}$ of a codeword that maximizes the cosine similarity
\begin{align}
    i^* \in  \arg \max_{1\leq i \leq M}  \langle C(i) ,x\rangle.
\end{align}
Given the index $i$ and knowledge of the codebook $\Ccal$, the decoder outputs the compressed representation 
\begin{align}
    Y \coloneqq \rho\, C(i^*),
\end{align}
where $\rho \ge 0$ is a scaling parameter. 
\end{definition}

Our results focus on the distribution of the compressed representation $Y$ induced by the randomness in the codebook. Note that this distribution is parameterized by the input $x$ and has a density with respect to the surface measure on the sphere of radius $\rho$. 
% has a 
% Considering the randomness in the codebook, the distribution of $Y$ has a density with respect to the surface measure on the sphere of radius $\rho$ in $\reals^n$ and $\cos(\alpha^*)$ is a random variable taking values in $[-1,1]$.
We also define the maximal cosine similarity according to
\begin{align}
    \label{eq:alpha_star}
\cos(\alpha^*) \coloneqq \frac{\langle x, C(i^*) \rangle}{\|x\|},\qquad \alpha^* \in [0, \pi].
\end{align}
It is well known (see e.g.,\cite{stam:1982}) that the distribution of  $\alpha^*$ does not depend on $x$ and is given by
%It follows from \cite[Eq. (3)]{stam:1982} that 
%$\cos(\alpha^*)$ has the distribution 
\begin{align}
\label{eq:alpha_CDF}
    \pr{ \left(\cos(\alpha^*) \le s  \right)}  = \left(1-Q_n(s) \right)^M,
\end{align}
where 
\begin{align}
\label{eq:Qn}
Q_n(s) \coloneqq
%\frac{\mu_{n-1}}  {\sqrt{2\pi}}
\frac{\Gamma(\frac{n+1}{2})}{\sqrt{\pi}  \Gamma(\frac{n}{2})}
\int_s^1 (1 - t^2)^{\frac{n-3}{2}} \, dt,%\quad \mu_n =\frac{\sqrt{2}\Gamma(\frac{n+1}{2})}{ \Gamma(\frac{n}{2})}
\end{align}
and where $\Gamma(z) = \int_0^\infty x^{z-1} \exp ( - z) \, dz$ is the Gamma function.
%\par
% Considering the randomness in the codebook, the distribution of $Y$ has a density with respect to the surface measure on the sphere of radius $\rho$ in $\reals^n$ and $\cos(\alpha^*)$ is a random variable taking values in $[-1,1]$.

\begin{figure}
    \begin{center}
\usetikzlibrary{decorations.text}
\begin{tikzpicture}
\def\s{3};
\def\per{.3};
\def\ang{70};
\def\angr{0};
\def\r{(\s) / sin(\ang)};

\node (x) at (0,{\s}) {};
\node (xrho) at (0,{\r * sin(\ang)}) {};
\node (Y) at ({\r*cos(\ang)},{\r*sin(\ang)}) {};
\node (rho) at ({\r*cos(\angr)},{\r*sin(\angr)}) {};

\draw[decoration={text along path,
      text={representation sphere}, text align={center}},decorate, yshift=.15cm] (0,{\r}) arc (90:30:{\r});

\draw (0,{\r}) arc (90:0:{\r});

\draw[-|,dashed, line width = 1pt] (0,0) -- node[above, rotate = 0, xshift = -0.0cm]{$\rho$} ({\r*cos(\angr)},{\r*sin(\angr)});
  
\draw (0,{\s/3}) arc (90:70:{\s/3}) node[above,xshift=-.1cm] {\small $\alpha^*$}; 
  
 \draw[color = blue,thick] (x.center) --  (Y.center);
  
\draw[color = red,line width = 1pt,->] (0,0) -- node[above, xshift = -0.2cm]{$x$} (0,\s);

\draw[|-|,color = gray,line width = .5pt,dashed] ({\r * cos(\ang)},0) -- node[right] {\scriptsize $\rho \cos(\alpha^*)$} (Y.center);

\draw[-|,color = gray,line width = .5pt] (0,0) -- node[below] {\scriptsize $\rho \sin(\alpha^*)$} +({\r * cos(\ang)},0);

\draw[color = black, line width = 1pt,->] (0,0) -- node[above, xshift = 0.35cm, yshift = -0.1cm]{$Y$} (Y.center);
\end{tikzpicture}
\begin{tikzpicture}

\def\s{3};
\def\per{.3};
\def\ang{65};
\def\angr{0};
\def\r{(\s+\per) / sin(\ang)};

\node (x) at (0,{\s}) {};
\node (xrho) at (0,{\r * sin(\ang)}) {};
\node (Z) at ({\r*cos(\ang)},{\r*sin(\ang)}) {};
\node (rho) at ({\r*cos(\angr)},{\r*sin(\angr)}) {};

%\draw[-|,dashed, line width = 1pt] (0,0) -- node[above, rotate = 20, xshift = -0.0cm]{$\rho$} ({\r*cos(\angr)},{\r*sin(\angr)});
  
\draw[color = blue, thick] (x.center) -- node[above, yshift=0cm] {\small $\sigma W$}  (Z.center);
  
%\node at ({\r},{\r}) {\small $W \sim \Ncal(0,I_n)$};
  
\draw[color = red,line width = 1pt,->] (0,0) -- node[above, xshift = -0.2cm]{$x$} (0,\s);

\draw[|->,color = gray,line width = .75pt] (x.center) -- node[left] {\scriptsize $\sigma A$} (xrho.center);

\draw[->,color = gray,line width = .75pt] (0,0) -- node[below] {\scriptsize $\sigma B$} +({\r * cos(\ang)},0);

\draw[dashed,color = gray,line width = .5pt] ({\r * cos(\ang)},{\s}) -- ({\r * cos(\ang)},0);

\draw[color = black, line width = 1pt,->] (0,0) -- node[above, xshift = 0.35cm, yshift = -0.1cm]{$Z$} (Z.center);
\end{tikzpicture}
    \end{center}
    \caption{
    Conceptual 2D description of the coupling in Theorem~\ref{thm:coupling} when the representation sphere is matched to $\|x\|$. 
    The quantization error $Y-x$ (Left) is compared to the standard $n$-dimensional Gaussian noise vector $W$ (Right). The random variable $Y$ is the nearest codeword to $x$ in the random codebook ensemble. 
    %The random angle $\alpha^*$ depends is approximated by $\sin(\alpha^*) \sim (M/\sqrt{n})^{-1/(n-1)}$. 
    The standard Gaussian random variable $A$ is the normalized component of $W$ in the direction of $x$. The random variable $B$ describes the magnitude of the projection of $W$ onto the $(n-1)$-dimensional space orthogonal to $x$. 
    }
    \label{fig:coupling}
\end{figure}
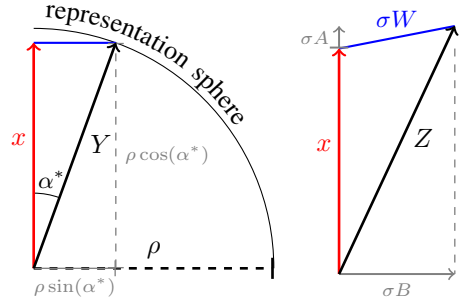

\subsection{Approximation using AWGN}
{
The fundamental question we address is the extent to which the quantization error $Y - x$ can be approximated by an isotropic zero-mean Gaussian noise. To answer this question we introduce the AWGN-corrupted observation model
\begin{align}
\label{eq:Z_channel_def}
Z = x + \sigma W, \qquad W \sim \normal(0, I_n).
\end{align}
Our main results are based on a coupling argument. Specifically, we show that there exists a joint distribution on the pair $(Y,Z)$ under which the distribution of $\|Y-Z\|$ can be described exactly in terms of the tuple $(\rho, \sigma, n,M)$ and the magnitude of the input. The proof of the following result is in Section~\ref{sec:proof:coupling}.

\begin{thm}
\label{thm:coupling}
For any $x\in \reals^n \backslash\{0\}$,  positive integer $M$, and real numbers $\rho,\sigma > 0$, there exists a joint distribution on $(Y,Z)$ such that $Y$ has the distribution of the compressed representation of magnitude $\rho$ obtained from an $(n,M)$ random spherical code with input $x$, $Z \sim \Ncal(x, \sigma^2)$, and 
\begin{align}
    \label{eq:coupling}
\left\| Y- Z \right\|^2 = \left(\rho \cos(\alpha^*) - \|x\| - \sigma A \right)^2 + \left( \rho \sin(\alpha^*) - \sigma B \right)^2,
\end{align}
where:
\begin{itemize}
    \item $A,B,\alpha^{*}$ are independent,
    \item $A\sim \Ncal(0,1)$,
    \item $B$ has a chi distribution with $n-1$ degrees of freedom.
    \end{itemize}
\end{thm}

Figure~\ref{fig:coupling} provides a conceptual illustration of $\alpha^*$, $A$, and $B$ in the comparison between $Y$ and $Z$ provided in Theorem~\ref{thm:coupling}. The proof of this theorem in  Section~\ref{sec:main_proof} provides the exact description of these random variables and vectors. 

{
The coupling described in Theorem~\ref{thm:coupling} holds for any choice of the parameters $(\rho, \sigma)$. The next step is to show that the term $\|Y-Z\|$ in \eqref{eq:coupling} is negligible compared to the magnitude of the quantization error under a proper specification of these parameters. To provide a sense of scale, observe that the error due to AWGN satisfies  $\|Z - x\| = \sigma \|W\|$ where $\|W\|$ concentrates about $\sqrt{n}$ in  the large-$n$ limit. For comparison, we consider the upper bound given by
\begin{align}
\left\| Y- Z \right\| \le |\|x\| - \gamma| +  \sigma \Delta \label{eq:coupling_UB}
\end{align}
where $\gamma$ is an estimate of $\|x\|$ and the normalized error term %$\Delta$ is a normalized error term defined by
\begin{align}
\label{eq:Delta}
    \Delta \coloneqq \frac{1}{\sigma} \sqrt{  \left(\rho \cos(\alpha^*) - \gamma - \sigma A \right)^2 + \left( \rho \sin(\alpha^*) - \sigma B \right)^2} %+ \sqrt{A^2 + (B - \sqrt{n})^2}
\end{align}
does not depend on $x$. In the following we show that $(\rho, \sigma)$ can be chosen as a function of $(n,M, \gamma)$ such that the distribution of $\Delta$ is bounded independently of the dimension $n$. %Under the assumption $\|x\|/\gamma \to 1$ as $n \to \infty$, this means that the relative difference between $\|Y-Z\|/\|Y-x\|$ converges to zero. 
%When this occurs, it means that $\Delta/\|W\|$ converges to zero at rate $n^{-1/2}$. 
}

First we consider the setting where the number of codewords is given by $M = \lceil 2^{nR} \rceil$ for a fixed bitrate $R>0$. For a given $\gamma \ge 0$ we use the specification 
\begin{align}
\label{eq:rhosig}
    \rho = \frac{ \gamma }{ \sqrt{ 1 - 2^{-2R}}}, \qquad \sigma = \frac{\gamma}{ \sqrt{n} \sqrt{ 2^{2R}-1}}.
\end{align}

%The next result provides a sub-Gaussian bound on the relative error term. 
 
\begin{thm}
\label{thm:DeltaR}
Suppose that $M = \lceil 2^{nR} \rceil$ for a fixed bitrate $R>0$. Under the specification given in  \eqref{eq:rhosig},
the normalized error term $\Delta$ defined in \eqref{eq:Delta} has a sub-Gaussian distribution with parameters that depend only on the bitrate $R$. In particular, there exists a positive number $C_R$ such that
\begin{align}
\label{eq:thm:Delta}
    \ex{ \Delta^p}^{1/p} \le C_R \sqrt{p}, \qquad p \ge 1. 
\end{align}

\end{thm}

The significance of Theorem~\ref{thm:DeltaR} is that the distribution of $\Delta$ is bounded uniformly with respect to $n$ and thus the term $\sigma \Delta$ in \eqref{eq:coupling_UB} is order one. By comparison, the magnitude of the additive noise $\sigma \|W\|$  scales at rate $n^{1/2}$. This means that if the estimate of $\|x\|$ is accurate in the sense that $\|x\|/\gamma \to 1$ as $n \to \infty$, then the relative difference between the mismatch  $\|Y-Z\|$ and the 
quantization error $\|Y-x\|$ converges to zero.

%Under the assumption  $\|x\|/\gamma \to 1$ as $n \to \infty$, $\sigma \Delta$ is order one, and thus asymptotically negligible compared to the magnitude of the quantization error, which scales with $n^{1/2}$. An explicit value for the constant $C_R$ and its dependence on $R$ can be found in the proof of Theorem~\ref{thm:DeltaR}, which is given in Appendix~\ref{sec:Proof2_3}.\par
%

%One limitation of Theorem~\ref{thm:DeltaR} is that dependence on the bitrate $R$ and thus it does not say whether the difference between $|Y-Z|$ and the quantization error is small in the high-rate regime
{
Next, we provide a result that holds in the high-rate setting where $\frac{1}{n} \log M$ diverges. This regime requires a more precise estimate of the max-cosine similarity and we use the specification
\begin{align}
\label{eq:rhosig_tilde}
  \tilde{\rho} = \frac{ \gamma }{ \sqrt{ 1 - 2^{-2\tilde{R}}}},~~ \tilde{\sigma} = \frac{\gamma}{ \sqrt{n} \sqrt{ 2^{2\tilde{R}}-1}}, ~~ \tilde{R} = \frac{1}{n-1} \log_2 M.
\end{align}
Note that $\tilde{R}$ is normalized by $(n-1)$ instead of $n$. Further, we will  require that the number of codewords is bounded from below by
\begin{align}
\label{eq:thm:Delta_tilde}
     M_\beta(n)\coloneqq \sqrt{n} (\csc \beta)^{n-1},
\end{align}
for some fixed constant $\beta \in (0, \pi/2)$. 

\begin{thm}\label{thm:Delta_beta}
Suppose that $M \ge M_\beta(n)$ for some  constant $\beta \in (0, \pi/2)$. Under the specification given in \eqref{eq:rhosig_tilde}, the normalized error term $\Delta$ defined in \eqref{eq:Delta} has a sub-exponential distribution with parameters that depend only on $\beta$. In particular, there exists a positive number $C_\beta$  such that
\begin{align}
    \ex{ \Delta^p}^{1/p} \le C_\beta \,  p, \qquad p \ge 1. 
\end{align}
\end{thm}

 Theorem~\ref{thm:Delta_beta} is stronger than  Theorem~\ref{thm:DeltaR} in the sense that the bound holds uniformly for all $(n,M)$ satisfying the constraint $M\ge M_{\beta}(n)$. This is important for the high-bitrate setting where $\sigma$ converges to zero. The price that is paid is that the sub-Gaussian tail condition is replaced with the weaker sub-exponential condition. An explicit value for the constant $C_\beta$ and its dependence on $\beta$ can be found in the proof of Theorem~\ref{thm:Delta_beta}, which is given in Appendix~\ref{sec:Proof2_3}.
}
}
\subsection{Bounds on Wasserstein Distance} 
Our results can also be stated in terms of the Wasserstein distance on distributions. The $p$-Wasserstein distance between distributions $P$ and $Q$ on $\reals^n$ is defined by
\begin{align*}
W_p(P,Q) \coloneqq \inf \left({\ex{ \| U - V\|^p}} \right)^{1/p},
\end{align*}
where the infimum is over all joint distributions on $(U,V)$ satisfying the marginal constraints $U \sim P$ and $V\sim Q$. For $p \geq 1$, the $p$-Wasserstein distance is a metric on the space of distributions with finite $p$-th moments. \par
Theorems~\ref{thm:DeltaR} and \ref{thm:Delta_beta} imply upper bounds on the Wasserstein distance between the distribution of the compressed representation obtained using a random spherical code and the distribution of the AWGN-corrupted version of the input. 
\begin{thm}\label{thm:Wpbound}
Let $X$ be a random vector in $\reals^n$ with $\ex{ \|X\|^p} < \infty$ for some $p \ge 1$. Let $P_Y$ be the distribution of the compressed representation of magnitude $\rho$ obtained from an $(n,M)$ random spherical code with input $X$ 
%distribution of the output of a random spherical code $(n,M)$ with input $X$ and magnitude $\rho$, 
and let $P_Z$ be the distribution of $Z = X + \sigma W$ where $W$ is an independent standard Gaussian vector.  
\begin{itemize}
    \item [(i)] If $(\rho,\sigma)$ are defined as in \eqref{eq:rhosig}, then
\begin{align}
\label{eq:YZ_p_bound}
W_p(P_Y, P_Z) \le \left( \ex{ \left| \|X\| - \gamma \right|^p} \right)^{1/p} + \sigma C_R \sqrt{p},
\end{align} 
where $C_R$ is a positive number that depends only on $R$.
\item [(ii)] If $M \ge M_n(\beta)$ for some $\beta \in  (0, \pi/2)$ and $(\rho,\sigma)$ are defined as in \eqref{eq:rhosig_tilde}, then
\begin{align}
\label{eq:YZ_p_bound_tilde}
W_p(P_Y, P_Z) \le \left( \ex{ \left| \|X\| - \gamma \right|^p} \right)^{1/p} + \sigma  C_\beta p,
\end{align} 
where $C_\beta$ is a positive number that depends only on $\beta$.
\end{itemize}
\end{thm}
\begin{proof}
The $p$-th power of the Wasserstein distance is convex in the pair $(P,Q)$ \cite[Theorem~4.8]{villani2008optimal}, and thus
\begin{align}
    W^p_p(P_Y,P_Z)& \le  \int W^p_p(P_{Y\mid X=x},P_{Z \mid X = x}) \, \dd P_{X}(x),
\end{align}
where $P_{Y\mid X=x}$ and $P_{Z \mid X =x}$ denote the conditional distributions of $Y$ and $Z$, respectively. In view of \eqref{eq:coupling_UB} and \eqref{eq:thm:Delta}, it follows that
\begin{align}
 W_p(P_{Y\mid X=x},P_{Z \mid X = x}) \le \left| \|x\| - \gamma \right|  + \sigma C_R \sqrt{p}.
\end{align}
Combining these displays with Minkowski's inequality leads to \eqref{eq:YZ_p_bound}. Inequality~\eqref{eq:YZ_p_bound_tilde} is obtained in a similar manner from \eqref{eq:thm:Delta_tilde}.
\end{proof}

A useful property of the Wasserstein distance is that it controls the expectations of Lipschitz continuous functions. Recall that a mapping $f: \reals^n \to \reals^m $ is Lipschitz continuous if there exists a constant $L$ such that 
\begin{align}
   \|f(u)- f(v)\|  \le L \|u - v\|, \quad \text{for all $u,v \in \reals^n$.}
\end{align}
 The infimum over all $L$ is call the Lipschitz constant and is denoted by by $\|f\|_\mathrm{Lip}$. The 1-Wasserstein distance, which is also known as the Kantorovich–Rubinstein distance,  can be expressed equivalently as
 \begin{align}
W_1(P,Q) = \sup\left\{ \bEx\left[f(U) \right] - \bEx\left[ f(V) \right] \, \mid \, \| f\|_\mathrm{Lip} \le 1 \right\},
 \end{align}
 where $U \sim P$ and $V \sim Q$.  More generally, the $p$-Wasserstein can be used to bound the difference between  $p$-th moments. 
 
\begin{prop}\label{prop:Wp_cont}
Let $U \sim P$ and $V \sim Q$ be random vectors on $\reals^n$. For any Lipschitz function $f : \reals^n \to \reals^m$,
\begin{align}
\left| \left( \mathbb{E} \left[ \|f(U)\|^p\right]\right)^{1/p} - \left( \mathbb{E} \left[ \|f(V)\|^p\right]\right)^{1/p}\right| \le   \|f\|_\mathrm{Lip}\, W_p(P,Q), \label{eq:W_cont}
\end{align}
provided that the expectations exist.
\end{prop}
\begin{proof}
For any coupling of $(U,V)$, Minkowski's inequality and the Lipschitz assumption on $f$ yield
\begin{align*}
%\MoveEqLeft 
\left(\ex{ \|  f(U)\|^p} \right)^{1/p} 
& \le \medmath{ \left(\ex{ \|  f(V)\|^p} \right)^{1/p} + \left(\ex{ \| f(U) - f(V)\|^p}\right)^{1/p}}\\
&\le \medmath{\left(\ex{ \|  f(V)\|^p} \right)^{1/p} + \|f\|_\mathrm{Lip} \left(\ex{ \|U - V\|^p}\right)^{1/p}}.
\end{align*}
Taking the infimum over all possible couplings leads to one side of the inequality.  Interchanging the role of $U$ and $V$ and repeating the same steps gives the other side.
\end{proof}

%\nr{GR: This is essentially a corollary of our earlier Theorem, right?} AK: right

\subsection{Concentration of the Norm}
\label{sec:concentration}

The bounds in Theorems~\ref{thm:DeltaR} and \ref{thm:Delta_beta} simplify further when the parameter $\gamma$ in \eqref{eq:Delta} is matched to the magnitude of the input. %
As a specific example, suppose that the data is known to lie on the sphere of radius $\sqrt{n}$ and that $M=\lceil2^{nR} \rceil$. By setting $\gamma = \sqrt{n}$, we see that there exists a coupling of $Y$ and $Z$ under which the quantization error $\|Y-Z\|$ is bounded by a sub-Gaussian random variables that is independent of $n$. \par
For many applications, however, the assumption that the data lay on a sphere of known radius is too restrictive. Therefore, in this paper, we assume that the data at the input to the compressor is a random vector $X$ in $\reals^n$ whose magnitude $\|X\|$ concentrates about a known value $\gamma$. 
This assumption is reasonable for high-dimensional settings where the entries of $X$ are weakly correlated. More generally, there are a number of other approaches that can be used to deal with the fact that $\|X\|$ is unknown. One approach is to use additional bits to encode the magnitude of $X$, as is done in \cite{zhu2014quantized}. For example, if $\|X\| \le  \kappa \sqrt{n}$ almost surely where $\kappa$ is a known constant, then $\log_2 \sqrt{n}$ bits are sufficient to encode $\|X\|$ with absolute error less than $\kappa$, such that 
\[
\left( \ex{ \left| \|X\| - \gamma \right|^p} \right)^{1/p} \leq \kappa.
\]
When $n$ is large, the logarithmic number of bits used to encode the magnitude of $X$ is negligible compared to the $n R$ bits used to encode its direction. An alternative approach is to compare the compressed representation with a noisy version of $X$ after it has been projected onto the unit sphere in $\reals^n$. This can be achieved, by setting $\gamma = 1$ and redefining the input to be $\tilde{X} =  X/\|X\|$ such that the magnitude is equal to one almost surely. 
%In this setting, the term right-hand side of \eqref{eq:YZ_p_bound} becomes $C \sqrt{p/n}$. 
In both of the approaches described above, the noise variance $\sigma^2$ is scaled in such a way that the signal-to-noise ratio in the AWGN observation model \eqref{eq:Z_channel_def} depends only on $(n,R)$ and is given by
$(\rho/\sigma)^2$. One may also consider a variable-length coding strategy that adapts the number of bits to the magnitude of $X$ such that the effective noise power is constant and the signal-to-noise ratio is proportional to $\|X\|^2$. We leave this as a direction for future work.

\section{Application to Parameter Estimation \label{sec:parameter_estimation}} 

In this section, we apply our main results to the problem of estimating an unknown parameter vector $\theta$ from a compressed representation of the data $X$. 
For each integer $n$, let $\cP_n = \{P_{n,\theta}\, : \, \theta \in \Theta_{n}\}$ be a family of distributions on $\reals^n$ with index set $\Theta_n \subseteq \reals^{d_n}$. For the purposes of exposition we will focus on the squared error loss. Our approach is quite general, however, and can be extended to other loss functions. \par
An important performance benchmark in estimating $\theta$ from a bitrate-$R$ compressed representation of $X$ is the minimax MSE:
\begin{align}
\cM_{n}^*  \coloneqq
\inf_{\phi,\psi}  \sup_{\theta \in \Theta_n} \frac{1}{d_n}  \bEx_{P_{n,\theta}}\left[  \| \theta  - \psi(\phi(X))\|^2 \right] ,
\end{align}
where the minimum is over all encoding functions $\phi : \reals^n \to \{1, \dots, M\}$ and decoding functions $\psi: \{1, \dots, M\} \to \reals^d$ with $M = \lceil 2^{nR} \rceil$. 

Zhu and Lafferty~\cite{zhu2014quantized} studied the asymptotic minimax MSE for the Gaussian location model $X \sim \Ncal(\theta, \eps^2 I_n)$ with $\Theta_n$ the $n$-dimensional Euclidean ball of radius $\kappa \sqrt{n}$, and showed that
\begin{align}
\limsup_{n \to \infty} \cM_n^* = \frac{ \kappa^2 \eps^2}{\kappa^2 + \eps^2} +\frac{ \kappa^4}{ \eps^2 + \kappa^2} 2^{-2R}.
\label{eq:linear_minimax_risk}
\end{align}
%Note that the first term is the minimax mean-squared error without any rate constraint (Pinsker's Theorem). 
Their achievability result is based on random spherical coding while devoting a number of bits sublinear in $n$ to encode the magnitude of the $X$, as discussed in Section~\ref{sec:concentration} above.
\par
The comparison between quantization error and Gaussian noise in Theorem~\ref{thm:Wpbound} provides a straightforward method for obtaining non-asymptotic upper bounds on the minimax MSE that can be applied to a large class of models. The basic idea is to study the MSE of Lipschitz estimators applied to the AWGN-corrupted data. We use the following assumption, which says that $P_{n,\theta}$ concentrates  on a spherical shell whose radius does not depend on $\theta$. 

\begin{assumption}[Concentration of Magnitude]\label{assumption:near_sphere} %Let  $\{P_\theta \, : \, \theta \in \Theta\}$ be a family of distributions on $\reals^n$ with index set $\Theta \subset \reals^d$.
There exists a sequence of positive numbers $\{ (\gamma_n, \tau_n)\}_{n \in \mathbb{N}}$ such that %$\gamma_n$
\begin{align}
\sup_{\theta \in \Theta_n}\bEx_{ P_{n,\theta}} \left[\left| \|X\| - \gamma_n \right|^2 \right] \le  \tau^2_n
\end{align}
\end{assumption}
Assumption~\ref{assumption:near_sphere} provides a way to formulate many cases of interest in terms of the radius of the shell $\gamma_n$ and its width $\tau_n$. \par
The next result uses this assumption to bound the difference in root MSE between an estimator applied to the compressed representation $Y$ and the same estimator applied to the AWGN-corrupted version $Z$. 
\begin{thm}\label{thm:parameter_estimation}
Let $\{P_{n,\theta} \}_{n \in \mathbb{N}}$ be  a sequence of models that satisfies Assumption~\ref{assumption:near_sphere}. Given $X \sim P_{n,\theta}$, let $Y$ be the output of an $(n,\lceil2^{n R} \rceil)$ random spherical code with input $X$ and output scaling $\rho$ and let $Z= X + \sigma W$ where $W \sim \normal(0,I_n)$ is independent of $X$ and $\rho$ is according to \eqref{eq:rhosig}. For any Lipschitz estimator $\hat{\theta} : \reals^n \to \reals^{d_n}$, the root MSE satisfies
\begin{align} \label{eq:thm_parameter}
\left| \sqrt{ \mathbb{E}\left[ \| \hat{\theta}(Y) - \theta\|^2\right]} - \sqrt{ \mathbb{E} \left[ \| \hat{\theta}(Z) - \theta\|^2 \right] } \right| 
\le C\,  \| \hat{\theta}\|_\mathrm{Lip}\, \beta_n
\end{align}
where $C$ is a constant that depends only on the bitrate $R$ and $\beta_n = \tau_n \vee (\gamma_n / \sqrt{n})$. Furthermore, for all $t > 0$, the minimax MSE satisfies
\begin{align} 
    \cM^*_n
    \leq 
    \frac{1}{n} \ex{ \| \hat{\theta}(Y) - \theta\|^2}
    & \le (1 + t) 
  \frac{1}{d_n} \ex{  \| \hat{\theta}(Z) - \theta\|^2 }
  \label{eq:thm_parameter2}
  \\
  & \quad +  \left(1 + \frac{1}{t} \right) C^2 \frac{ \| \hat{\theta}\|^2_\mathrm{Lip} \beta^2_n}{d_n}.  
  \nonumber
\end{align}
\end{thm}

\begin{proof}
For each $\theta \in \Theta_n$, let $Q_{n,\theta}$ and $Q'_{n,\theta}$ denote the corresponding distributions of $Y$ and $Z$. Proposition~\ref{prop:Wp_cont} evaluated with $f(u) = \| \hat{\theta}(u) - \theta\|$ gives, 
\begin{align*}
 & \left| \left( \mathbb{E}  \left[ \| \hat{\theta}(Y) - \theta\|^2\right]\right)^{1/2} - \left( \mathbb{E} \left[ \| \hat{\theta}(Z) - \theta\|^2 \right] \right)^{1/2} \right| \\
 & \qquad \le  \|\hat{\theta} \|_\mathrm{Lip} W_2(Q_{n,\theta}, Q'_{n,\theta}),
\end{align*}
where we have used the fact that $f$ is the composition of $\hat{\theta}$ with the 1-Lipschitz function $\|\cdot - \theta\|$, and thus  $\|f\|_\mathrm{Lip} = \|\hat{\theta}\|_\mathrm{Lip}$. 
By Theorem~\ref{thm:Wpbound} and  Assumption~\ref{assumption:near_sphere}, the Wasserstein distance is upper bounded  by $\tau_n + \sigma\sqrt{2}C_R$, where $\sigma$ is given in \eqref{eq:rhosig}. It follows that
\begin{align*} 
    \frac{1}{n} \ex{ \| \hat{\theta}(Y) - \theta\|^2}
    & \le 
    \left(
  \frac{1}{d_n} \ex{  \| \hat{\theta}(Z) - \theta\|^2 } + C'  \frac{\gamma}{\sqrt{n}}  \|\hat{\theta}\|_\mathrm{Lip}  
  \right)^2,
\end{align*}
where $C'$ only depends on $R$. The upper bound on the minimax MSE follows from the inequality $(a + b)^2 \le a^2 (1 + t)  + b^2 (1 + 1/t)$ for all $t > 0$.
\end{proof}

One takeaway from Theorem~\ref{thm:parameter_estimation} is that the MSE obtained from the compressed representation is asymptotically equivalent to that of the AWGN-corrupted observation, provided that the Lipschitz constant of the estimator is small enough. 
To gain insight into the interplay between the Lipschitz constant of the estimator, the magnitude of the data, and the typical size of the squared error, it is useful to consider some concrete examples. 

\subsection{Gaussian Location Model} 
For our first example, we consider the Gaussian location model $X \sim \normal(\theta, \eps^2 I_n)$. % studied in \cite{zhu2014quantized}.
Assume that the parameter set $\Theta_n$ is a subset of the spherical shell:
\begin{align}
 \cS_n \coloneqq   \{ \theta \in \reals^n \, : \,  \kappa \sqrt{n} - \omega_n \le  \|\theta\|  \le \kappa \sqrt{n} + \omega_n \}. 
\end{align}
As an intuition for this notation, one may think about $\sqrt{n}\kappa$ as an estimate for the magnitude of $\theta$ and $\omega_n$ as the uncertainty in this estimate. 
For example, if the entries of $\theta$ are sampled independently from a sub-Gaussian distribution with a second moment $\kappa^2$, then $\|\theta\|-\sqrt{n}\kappa$ is sub-Gaussian \cite[Thm 3.1.1]{vershynin2018high}. %\nr{GR: which result is this?}
In this case, there exists a constant $C$ independent of $n$ such that 
$\theta \in \cS_n$ for $\omega_n = C\sqrt{2 \log n}$ with probability at least $1-1/n$.\par
%\nr{GR: Need to decide if shell has $\kappa /\sqrt{n}$ as upper bound or $\kappa \sqrt{n} + \omega_n$. If we keep it as is, then the example needs to be adjusted so that entries have second moment less than $\kappa^2$.}
%\nb{AK: I have changed the upper bound.}. 
%
The next statement says that under a Gaussian location model, $X$ concentrates whenever $\theta$ is restricted.
\begin{prop} \label{prop:concentration_theta}
Consider the model $X \sim \normal(\theta, \eps^2 I_n)$ with $\Theta_n \subseteq \cS_n$.  Assumption~\ref{assumption:near_sphere} is satisfied with $\gamma_n = \sqrt{n(\kappa^2 + \eps^2)}$ and $\tau_n = \omega_n + 2 \eps$. 
\end{prop}
\begin{proof}
Let $\mu = \ex{\|X\|^2} = \|\theta\|^2 + n \eps^2$. 
By the triangle inequality,
\begin{align}
\label{eq:omega_proof}
    \sqrt{\ex{  \left| \|X\| - \gamma_n \right|^2 } } & 
    \le     \sqrt{\ex{  \left| \|X\| - \sqrt{ \mu} \right|^2 } } + \left| \sqrt{\mu} - \gamma_n \right|.
\end{align}
The assumption $\Theta_n \subseteq \cS_n$ implies
\begin{align*}
|\sqrt{\mu} - \gamma_n| & = |\sqrt{\|\theta\|^2 + n\epsilon^2} - \sqrt{n(\kappa^2 + \epsilon^2)}| \\
& = \left|\|\theta\| - \kappa\sqrt{n} \right|\frac{\|\theta\| + \sqrt{n}\kappa }{\sqrt{\|\theta\|^2 + n \epsilon^2} +\sqrt{n(\kappa^2 + \epsilon^2)}} \\
& \leq \left|\|\theta\| - \kappa\sqrt{n} \right| \leq \omega_n.
\end{align*}
Consequently, the second term on the right-hand side of \eqref{eq:omega_proof} is upper bounded by $\omega_n$. For the first term, we write
\begin{align*}
     & \ex{  \left| \|X\| - \sqrt{ \mu} \right|^2 } \le   \ex{  \left| \|X\| - \sqrt{ \mu} \right|^2 \left(1 + \frac{ \|X\|}{\sqrt{\mu}}\right)^2}  \\
     & \quad = \frac{ \Var(\|X\|^2) }{\mu} = \frac{2 \eps^2 (  2\|\theta\|^2 + n \eps^2)}{\|\theta\|^2 + n\eps^2} \le 4 \eps^2,
\end{align*}
where we have used the fact that $\|X\|^2/\eps^2$ has a  non-central chi-squared distribution with $n$ degrees of freedom and non-centrality parameter $\|\theta\|^2/\eps^2$
\end{proof}

In this setting, the AWGN-corrupted data $Z$, corresponding to a bitrate $R$ and magnitude $\gamma_n$, is drawn according to the Gaussian location model whose noise variance depends on the original noise level $\eps^2$ and the bitrate $R$: 
\begin{align} \label{eq:location_model_Z}
Z \sim \normal(\theta, (\eps^2 + \sigma^2) I_n), \qquad \sigma^2 = \frac{\kappa^2 + \eps^2}{2^{2R}-1}. 
\end{align}
The MSE in the Gaussian location model has been studied extensively. If we restrict our attention to linear estimators of the form $\hat{\theta}(z) = \lambda z$ then a standard calculation (see e.g., \cite[Ch. 4.8]{johnstone2011gaussian}) gives
\begin{align}
    %\inf_{\lambda \ge 0} \sup_{\theta \, : \, \|\theta\| \le \kappa n} \frac{1}{n} \ex{ \| \theta - \lambda Z\|^2} & = 
    \inf_{\lambda \ge 0} \sup_{\theta \, : \, \|\theta\| \le \kappa \sqrt{n}} \frac{1}{n} \ex{ \| \theta - \lambda Z\|^2} = \frac{\kappa^2 ( \eps^2 + \sigma^2)}{\kappa^2 + \eps^2 + \sigma^2},
    %&= \frac{ \kappa^2 \eps^2}{\kappa^2 + \eps^2} +\frac{ \kappa^4}{ \eps^2 + \kappa^2} 2^{-2R}.
\end{align}
where the minimum over $\lambda$ is attained at $\lambda^* = k^2/( \kappa^2 + \eps^2 + \sigma^2)$. 
By expressing the right-hand side as a function of $R$ and combining with Theorem~\ref{thm:parameter_estimation}, we obtain a non-asymptotic upper bound on the minimax MSE.
\begin{prop}
\label{prop:linear_minimax}
Consider the model $X \sim \normal(\theta, \eps^2 I_n)$ with $\Theta_n \subseteq \cS_n$. 
 Let $Y$ be the output of a bitrate-$R$ random spherical code applied to $X$ and scaled to the radius $\sqrt{n(\kappa^2+\epsilon^2)/(1-2^{-2R})}$. Then
\begin{align}\label{eq:prop_parameterd_linear}
\frac{1}{n}\ex{\left\|\theta - \lambda^* Y \right\|^2} \le \frac{ \kappa^2 \eps^2}{\kappa^2 + \eps^2} +\frac{ \kappa^4}{ \eps^2 + \kappa^2} 2^{-2R} + C\, \frac{1 \vee \omega_n}{\sqrt{n}},
\end{align}
where $C$ is a constant that depends on $(\kappa, \eps, R)$ but not  $n$. 
\end{prop}
\begin{proof}
We have $\gamma_n/\sqrt{n} = \sqrt{\kappa^2+\epsilon^2}$, and 
$\|\hat{\theta}\|_{\Lip}\leq 1$ for the linear estimator $\hat{\theta}(z)= \lambda^* z$. 
%\nr{GR: Why? what is $\hat{\theta}$?} \nb{AK: I added this detail}.
Following Proposition~\ref{prop:concentration_theta}, Assumption~\ref{assumption:near_sphere} is satisfied with $\tau_n = \omega_n + 2 \epsilon$. We use Theorem~\ref{thm:parameter_estimation} with $t=(1 \vee \omega_n)/\sqrt{n}$ and  $\beta_n = (\omega_n + 2 \epsilon)\vee \sqrt{\kappa^2+\epsilon^2}$. 
 %\nr{GR: What is the term  $\alpha$? I need more details to follow the derivation} \nb{AK: I added this detail}.
 It follows that there exists a constant $c$ such that 
\begin{align*}
\medmath{ \frac{1}{n}\ex{\left\|\theta - \lambda^* Y \right\|^2} } &  \medmath{ \le \left(1+\frac{1 \vee \omega_n}{\sqrt{n}}\right) \left( \frac{ \kappa^2 \eps^2}{\kappa^2 + \eps^2} +\frac{ \kappa^4}{ \eps^2 + \kappa^2} 2^{-2R} \right)} \\
& \quad  \medmath{ + \left( 1 + \frac{\sqrt{n}}{1\vee \omega_n} \right)c^2\, \frac{(\kappa^2+\epsilon^2) \vee (\omega_n+2\epsilon)^2}{n}}.
\end{align*}
Grouping $1/\sqrt{n}$ factors leads to \eqref{eq:prop_parameterd_linear}. 
\end{proof}
Since 
\[
\Mcal^*_n \leq \frac{1}{n}\ex{\left\|\theta - \lambda^* Y \right\|^2},
\]
Proposition~\ref{prop:linear_minimax} recovers parts of the results in \cite{zhu2014quantized} by showing that there exists a bitrate-$R$ coding scheme with minimax risk approaching \eqref{eq:linear_minimax_risk}. Furthermore, Proposition~\ref{prop:linear_minimax} shows that the minimax risk under such scheme converges at 
rate $1/\sqrt{n}$, which is faster than the convergence rate established in \cite{zhu2014quantized} by a factor of $1/\sqrt{\log n}$. \\
% We also note that $\Mcal^*$ describes the optimal trade-off between bitrate and MSE in estimating an i.i.d. Gaussian source from a bitrate-$R$ encoded version of its AWGN-corrupted observations \cite{DobrushinTsybakov,berger1971rate}. We discuss this trade-off in  Section~\ref{sec:source_coding} below. \\
%\nr{GR: Should we mention that here this result should be compared with the one of Zhu and Lafferty?}
%\nb{AK: We need more details to compare the two results. For example, Zhu and Lafferty assumed that $\|\theta\|\leq \kappa \sqrt{n}$ but not necessarily $\theta \in \Scal_n$. In a previous version we had this comparison but decided to remove it due to these extra details}. 

More generally, we can also provide bounds for non-linear estimators. The case of a $k$-sparse parameter vector can be modeled as \[
\Theta_n = \cS_n \cap \{ \theta \in \reals^n \, : \, \|\theta\|_0 \le k\}
\]
where $\|\theta\|_0$ denotes the number of nonzero entries in $\theta$. A great deal of work has studied the MSE of the soft-thresholding estimator 
\begin{align}
\hat{\theta}_\lambda(z) \coloneqq \begin{cases} z - \lambda, & z > \lambda, \\
0, & |z| \leq \lambda, \\
z+\lambda, & z < \lambda.
\end{cases}
\end{align}
in the model \eqref{eq:location_model_Z} \cite{Donoho1994,johnstone2011gaussian}. Specifically, we have %\cite[Thm 8.8 \& Lem. 8.3]{johnstone2011gaussian}
\begin{align} \label{eq:sparse_bound}
\adjustlimits \inf_{\lambda \ge 0} \sup_{\theta \, : \, \|\theta\|_0 \le  k}   \frac{1}{n} \ex{ \| \theta - \hat{\theta}_{\lambda }(Z)\|^2}  \le (\eps^2 + \sigma^2) \beta_0 \left(\frac{k}{n}\right),
\end{align}
where, for $\nu>0$, 
\begin{align}
    \beta_0(\nu) = \inf_{\lambda \ge 0} \medmath{\left\{  (1-\nu) [ 2 ( 1+ \lambda^2) \Phi(-\lambda)  - 2 \lambda \phi(\lambda)] + \nu (1+\lambda^2)\right\}}
\end{align}
where $\Phi(z)$ and $\phi(z)$ are the cumulative and density functions of the standard Gaussian distribution, respectively. Let $\lambda^*$ be the minimizer in \eqref{eq:sparse_bound}. 
\begin{prop} \label{prop:parameter_sparse}

Let $X\sim(\theta,\eps^2 I_n)$ where $\theta \in \Scal_n \cap \{ \theta\in \reals^n \,:\, \|\theta\|_0\leq k\}$. Let $Y$ be the output of a bitrate-$R$ random spherical code applied to $X$ and scaled to the radius $\sqrt{n(\kappa^2+\epsilon^2)/(1-2^{-2R})}$. Then 
\begin{align} \label{eq:sparse_minmax}
\frac{1}{n} \ex{\left\| \theta - \hat{\theta}_{\lambda^*}(Y)\right\|^2} \le  
\medmath{\left(\eps^2 + \frac{\kappa^2+\epsilon^2}{2^{2R}-1}\right) \beta_0 \left(\frac{k}{n} \right)} +  C \frac{1 \vee \omega_n}{\sqrt{n}},
\end{align}
where $C$ is a constant that depends on $(\kappa, \eps, R)$ but neither $k$ or $n$. 
\end{prop}
\begin{proof}
For any $\lambda >0$ we have $\|\hat{\theta}_{\lambda} \|_{\Lip} = 1$. 
Equation \eqref{eq:sparse_minmax} follows from Theorem~\ref{thm:parameter_estimation} by using
$t = (1\vee\omega_n)/\sqrt{n}$ and grouping $1/\sqrt{n}$ factors.
\end{proof}

\begin{cor}
\label{cor:sparse_regression}
Assume that $X\sim(\theta,\eps^2 I_n)$. The bitrate-$R$ constrained minimax risk  over $\theta \in \Scal_n \cap \{ \theta\in \reals^n \,:\, \|\theta\|_0\leq k\}$, with $\omega_n/\sqrt{n}\to 0$, satisfies 
\[
\Mcal_n^* \leq \left(\eps^2 + \frac{\kappa^2+\epsilon^2}{2^{2R}-1}\right) \beta_0 \left(\frac{k}{n} \right).
\]
\end{cor}

% \begin{figure}
%     \centering
%     \begin{tikzpicture}
%     \begin{semilogyaxis}[
%     legend pos = south west,
%     width=9cm, height=6cm,
%     xmin = 0, xmax = 3, 
%     ymin = 0, ymax=1, 
%     ytick = {0, 1},
%   % ylabel = {MSE},
%     xlabel = {$R$ (bitrate)},
%     yticklabels = {0,1},
%     xtick = {0,1,2,3},
%     xticklabels={0,1,2},
%     line width=1.0pt,
%     mark size=1.5pt,
%     ]
%     \def\del{0.4} %sparsity parameter
%     \def\Rmin{(-\del*log2(\del)-(1-\del)*log2(1-\del))}
    
%     \addplot[color = black, mark size=1.5pt, dashed, domain=\Rmin:3, samples = 27] {\del * 2^(-2*(x - (-\del*log2(\del)-(1-\del)*log2(1-\del))
%     )/\del)}; \addlegendentry{support encoding \cite{weidmann2012rate}};
    
%     \addplot[color = red, mark size=1.5pt, smooth, domain=0:3, samples = 21] {
%     \del*(2*ln(1/\del) +1.13/(ln(1/\del))^(0.5)+1)/(2^(2*x)-1)
%     }; \addlegendentry{spherical coding (this paper)};
    
%     \addplot+[mark=none,color=black, dotted] coordinates {(\Rmin, 1) (\Rmin, 0)};
    
% \end{semilogyaxis}
% \end{tikzpicture}
% \caption{Achievable MSE in encoding a Bernoulli-Gauss signal with a fraction of $\delta=0.4$ nonzero entries.
% }
% \end{figure}

\subsection{Linear Model with IID Matrix
\label{subsec:AMP}
} 
\newcommand{\AMP}{\mathsf{AMP}}

For the next example, we consider the linear model 
\begin{align}
    X \sim \normal(A \theta, \eps^2I_n),
\end{align}
where $A$ is a known $n \times d$ matrix,  $\theta$ is an unknown $d$-dimensional vector, and $\eps^2$ a known noise variance. In this setting, the AWGN-corrupted version of $X$ given by  $Z =X + \eps W$ with $W \sim \normal(0,\sigma^2 I_n)$ corresponds to a linear model with larger noise variance, that is 
\begin{align}
Z \sim \normal(A \theta, \xi^2 I_n), \qquad \xi^2 = \eps^2 + \sigma^2.
\end{align}

We study the approximate message passing (AMP) algorithm  \cite{donoho2009message} to estimate $\theta$ from $Z$. AMP is an iterative algorithm that can be defined by a sequence of scalar denoising functions $\{\eta_t\}_{t \ge 1}$ with $\eta_t : \reals \to \reals$ that are assumed to be Lipschitz continuous, and hence differentiable almost everywhere. Starting with an initial points $\hat{\theta}^0 = 0_{d \times 1}$ and $r^0 = 0_{n \times 1}$, a sequence of estimates $\hat{\theta}^t$ is generated according to 
\begin{align}
     \hat{\theta}^{t+1} &= \eta_t
     \left( A^\top r^t + \hat{\theta}^t \right), \label{eq:AMP1} \\
     r^t & =  Z - A \hat{\theta}^t + \frac{d}{n}r^{t-1}  \operatorname{div}\left( \eta_t( A^\top r^{t-1} + \hat{\theta}^{t-1}) \right) \label{eq:AMP2}
\end{align}
where $\eta_t(\cdot)$ is applied comopontwise and $\operatorname{div}(\eta_t(z)) = \frac{1}{n} \sum_{i=1}^n \eta_t'(z_i)$ with $\eta'_t(z) = \frac{\dd}{ \dd z} \eta_t(z)$. 

The main result of \cite{donoho2009message, BayatiMontanari2011} says that the MSE of each iteration of AMP can be characterized precisely in the high-dimensional limit when $A$ is a realization of a random matrix with i.i.d.\ zero-mean Gaussian entries. 
To formally state and use this result, we need to the following assumptions:
\begin{assumption}
\label{assumption:AMP1} 
$\{ \theta(n)\}_{n \in \mathbb N}$ is a sequence of $d_n$-dimensional vectors such that $n/d_n \to \delta \in (0,\infty)$ as $n$ goes to infinity. The empirical distributions of $\theta(1),\theta(2),\dots$, i.e., 
the probability distribution that puts a point mass $1/d_n$ at each of the $d_n$ entries of $\theta(n)$,
converges weakly to a distribution $\pi$ on $\reals$ with finite second moment $\kappa^2$. Furthermore, $\|\theta(n)\|^2/n$ converges to $\kappa^2$ as $n \to \infty$.
\end{assumption}

\begin{assumption}
\label{assumption:AMP2}
$\{ P_{\theta(n),n}\}_{n\in \mathbb N}$ is a sequence of models defined by $X\sim(A\theta(n), \epsilon^2 I_n)$, where the entries of $A$ are i.i.d. 
$\Ncal(0,1/n)$.  
\end{assumption}

For a fixed $n$, we further consider a sequence of estimators for $\theta(n)$ defined as follows:
\begin{assumption}
\label{assumption:AMP3}
$\{\eta_t\}_{t\in\mathbb N}$ is a sequence of scalar, Lipschitz continuous, and differentiable denoisers $\eta_t : \reals \to \reals$. For every $n,d_n \in \mathbb N$, the approximate message-passing (AMP) estimator $\theta^t_{\AMP}(z)$ is defined as the results of $t$ iterations of \eqref{eq:AMP1} and \eqref{eq:AMP2}. 
\end{assumption} 

The characterization of the MSE of the estimator $\theta^t_{\AMP}$ in the high-dimensional limit is given by the \emph{state evolution} recursion. This recursion is defined in terms of a distribution $\pi$ on $\reals$, sampling ratio $\delta\in(0,\infty)$, and initial noise level $\tau_0$, as
\begin{align}
    \label{eq:SE} 
    \tau_{t+1}^2 =
     \xi^2+ \frac{1}{\delta}\ex{ \left( \eta_t \left(\theta_0 + \tau_{t}W  \right) - \theta_0\right)^2 }  ,
     \quad t=1,2,\ldots,
\end{align}
where $\theta_0 \sim \pi$ and $W\sim\Ncal(0,1)$. Finally, define 
\begin{align*}
 \Mcal_{\mathsf{AMP}}^t(\xi^2) \coloneqq
\ex{ \left( \eta_t \left(\theta_0 + \tau_t W  \right) - \theta_0\right)^2 }, 
\end{align*}
where $\tau_t$ is given by $t$ iterations of \eqref{eq:SE}. Under assumptions \ref{assumption:AMP1}-\ref{assumption:AMP3} above, \cite[Thm. 1]{BayatiMontanari2011} implies that 
\begin{equation}
    \label{eq:BayatiMontanari_conclusion}
\lim_{n\to \infty} \frac{1}{d_n} \left\| \theta(n) - \theta^t_{\mathsf{AMP}}(Z) \right\|^2
= 
\Mcal^t_{\mathsf{AMP}}(\xi^2). 
\end{equation}

Combining this result with  Theorem~\ref{thm:parameter_estimation}, we conclude the following:
\begin{thm} \label{thm:AMP}
Consider a sequence of problems satisfying Assumptions~\ref{assumption:AMP1} and \ref{assumption:AMP2}. Let $Y$ be the output of an $(n,\lceil 2^{n R}\rceil)$ random spherical code applied to $X$ with radius $\rho$ for some $R>0$. Let $\theta^t_{\AMP}$ be an estimator satisfying Assumption~\ref{assumption:AMP3}. Then
\begin{align}
    \label{eq:thm_AMP}
\lim_{n\to \infty} \frac{1}{d_n} \ex{ \left\| \theta(n) - \theta^T_{\mathsf{AMP}}(Y) \right\|^2 \mid A} = \Mcal^t_{\mathsf{AMP}}(\xi^2_R),
\end{align}
almost surely, where
\[
\quad \xi^2_R = \epsilon^2 + \frac{\epsilon^2+\kappa^2/\delta}{2^{2R}-1}. 
\]
\end{thm}
\begin{proof}
Set $\gamma_n^2 = n(\epsilon^2+\kappa^2/\delta)$ and $\sigma^2 = \gamma_n/(n(2^{2R}-1))$. We first show that $X$ and $\gamma_n$ satisfy Assumption~\ref{assumption:near_sphere}. Since $A$ has i.i.d. entries $\Ncal(0,1/n)$, then $X\sim \Ncal\left(0,(\frac{1}{n} \|\theta(n)\|^2 + \epsilon^2)I_n\right)$. Using similar arguments as in Proposition~\ref{prop:concentration_theta}, we get
\begin{align*}
     & \left(\ex{ | \|X\| - \gamma_n|^2}\right)^{1/2} \\
     & \qquad  \le \left( \ex{ \left|\|X\|-\sqrt{\|A\theta(n)\|^2+n \epsilon^2} \right|^2} \right)^{1/2}
     + \omega_n.
\end{align*}
Assumption~\ref{assumption:AMP1} implies that $\theta = \theta(n) \in \Scal_{d_n}$ with $\omega_n = o(\sqrt{d_n})$. We conclude that
\begin{align*}
    & \ex{ \left|\|X\|-\sqrt{\|A\theta(n)\|^2+n \epsilon^2} \right|^2} \\
    & \qquad  \leq 
    \ex{ \frac{ 
    \left( \|X\|^2 - \|A\theta(n)\|^2 + n \epsilon^2 \right)^2}{ \|A \theta\|^2+n \epsilon^2 } } \\
    & \qquad \leq \frac{\Var(\|X\|^2)}{n \epsilon^2} = \frac{2n (\frac{1}{n}\|\theta(n)\|^2 + \eps^2)^2}{n\epsilon^2} \\ 
    & \qquad  \leq  
    \frac{2 (\frac{(\omega_n+\sqrt{d_n}\kappa)^2}{n} + \eps^2)^2}{\epsilon^2} = O(1),
\end{align*}
and thus Assumption~\ref{assumption:near_sphere} is satisfied for some $\tau_n = o(\sqrt{n})$. %
Let $L_{n,t} \coloneqq \|\theta^t_{\mathsf{AMP}}\|_{\Lip}$. In Appendix~\ref{app:AMP} we show that $\sup_n L_{n,t} < \infty$ almost surely. 
Applied to our setting, \eqref{eq:BayatiMontanari_conclusion} says that
\[
\left| \left(\frac{1}{d_n} \left\| \theta(n) - \theta^t_{\mathsf{AMP}}(Z) \right\|^2 \right)^{1/2}
- 
\sqrt{\Mcal^t_{\mathsf{AMP}}(\xi_R^2)} \right| = o(1).
\]
Using the triangle inequality once with the last display, Theorem~\ref{thm:parameter_estimation} implies that there exists $C$, that depends only on $R$ and $\kappa^2/\delta+\epsilon^2$, such that
\begin{align*}
 & \left| \left(\frac{1}{d_n} \ex{ \|  \theta(n)-  \theta^t_{\AMP}(Y)\|^2}\right)^{1/2} - \sqrt{\Mcal_{\AMP}^t(\xi_R^2)} \right| \\
 & \qquad \le \frac{C\, \| \theta^t_{\AMP} \|_\mathrm{Lip} \beta_n }{\sqrt{d_n}},
\end{align*}
with $\beta_n = o(\sqrt{d_n})$. 
\end{proof}

\section{Application to Indirect Source Coding 
\label{sec:source_coding}}

For the second application, we consider an indirect source coding setting where the observed data is a degraded version of the realization of an information source. The goal is to compress this version at bitrate $R$ and recover the source realization. Traditionally, both the encoder and decoder are designed with full knowledge of the joint distribution of the source and the data \cite{berger1971rate}. In this section, we study an encoding-decoding scheme where the encoder uses a random spherical code and the decoder is described by a Lipschitz estimator, which may be designed with partial or full knowledge of the distribution of the source and the data. Leveraging the results in Section~\ref{sec:main}, we show that the asymptotic performance can be described in terms of an AWGN model.\par
Throughout this section, the source and the data are modeled as a stochastic process $\{(U_n, X_n)\}_{n \in \mathbb{N}}$. The first $n$ terms in this sequence are denoted by $U^n = (U_1, \dots, U_n)$ and $X^n = (X_1 ,\dots, X_n)$. We focus on the squared error loss (or distortion function)  
\begin{align}
d(u^n, \hat{u}^n) = \frac{1}{n}\sum_{i=1}^n (u_i -\hat{u}_i)^2,
\end{align}
and assume the following regularity condition: 
\begin{assumption}\label{assumption:source_coding}
The process $\{(U_n, X_n)\}_{n \in \mathbb{N}}$  is stationary and second-order ergodic with finite second moments. In particular, this means that the empirical second moments converge in mean:
\begin{align}
    \frac{1}{n}\sum_{i=1}^n  \begin{pmatrix} U_i \\ X_i \end{pmatrix} \begin{pmatrix} U_i \\ X_i \end{pmatrix}^T \to \ex{ \begin{pmatrix} U_1 \\ X_1 \end{pmatrix} \begin{pmatrix} U_1 \\ X_1 \end{pmatrix} ^T }.
\end{align}
\end{assumption}

\subsection{The Indirect Distortion-Rate Function}
We begin by reviewing some basic properties of the indirect distortion-rate function, which describes the fundamental tradeoff between the bitrate $R$ and the expected distortion in our source coding setting. For each problem of size $n$, the indirect distortion-rate function is given by
\begin{align}
    D_n(R) & \coloneqq \min_{\phi, \psi}  \ex{ d\left( U^n ,  \psi( \phi(X^n))\right)},
\end{align}
where the minimum is over all encoding functions $\phi : \reals^n \to \{1, \dots, M\}$ and decoding functions $\psi: \{1, \dots, M\} \to \reals^n$ with $M = \lceil 2^{nR} \rceil$. The standard source coding setting corresponds to the special case where the source equals the data. When the source and data are stationary, as we assume in this paper, $n D_n(R)$ is sub-additive in $n$, and the limit 
\begin{align}
    D(R) \coloneqq \lim_{n \to \infty}  D_n(R)
\end{align}
is well-defined \cite[Lem. 10.6.2]{gray2011entropy}. \par
For some classes of processes, $D(R)$ can be expressed equivalently in terms of an optimization problem over a family of probability distributions subject to a mutual information constraint \cite{DobrushinTsybakov,berger1971rate}. Specifically, we have
\begin{align}
    D(R) = \lim_{n \to \infty} \min_{I(X^n ; \hat{U}^n) \le nR}   \ex{ d\left(U^n,  \hat{U}^n\right)},  \label{eq:DR_info}
\end{align}
where the minimum is over all joint distributions on $(U^n,X^n,\hat{U}^n)$ such that $(U^n,X^n)$ satisfy their marginal constraints, $U^n \to X^n \to \hat{U}^n$ forms a Markov chain, and $I(X^n; \hat{U}^n) \le n R$. 
For example, a representation of the form \eqref{eq:DR_info} exists for memoryless processes \cite{berger1971rate, 7464359} and in cases where the direct (standard) distortion-rate function of the sequence of random vectors $\tilde{U}^n = \ex{X^n \mid U^n}$ has a representation of the form \eqref{eq:DR_info}
by setting $X^n = U^n= \tilde{U}^n$ \cite[Ch. 3.2]{kipnis2017analog}. \par
There are a few cases where the distortion-rate function has simple closed-form expressions. For example if $\{(U_n,X_n)\}$ are i.i.d.\ from bivariate Gaussian distribution with zero mean, then the distortion-rate function is given by $D(R) = D_G(R)$ where 
\begin{align}
\label{eq:Gaussian_iDRF}
D_G(R) \coloneqq  \ex{|U_1|^2} -  \frac{\ex{ U_1 X_1}^2}{ \ex{ |X_1|^2}} \left( 1- 2^{-2R} \right).
\end{align}
This characterization was obtained in \cite{DobrushinTsybakov} and also \cite{1054469}. Note that the limiting case $R \to \infty$ corresponds to the minimum MSE in estimating $U_n$ from $X_n$. Moreover, for the direct source coding problem where $U_n$ is equal to $X_n$, this expression reduces to the standard distortion-rate function for an i.i.d.\ Gaussian source, $\ex{|X_1|^2} 2^{-2R}$.

\subsection{Achievability using Spherical Coding} 

We now consider the distortion that can be achieved when $X^n$ is compressed using a random spherical code. For each problem of size $n$, let $Y^n$ be the output of a bitrate-$R$ random spherical code with input $X^n$ and squared magnitude $n\ex{X_1^2} /(1 - 2^{-2R})$. The distortion-rate function associated with random spherical coding and estimator $f: \reals^n \to \reals^n$ is defined as 
\begin{align} \label{eq:Dspnf}
    \Dsp_n(R,f) & \coloneqq \ex{d\left( U^n, f(Y^n) \right)},
\end{align}
where the expectation is with respect to the joint distribution of $(U^n,Y^n)$. Under the squared error distortion, the minimum with respect to $f$ is achieved by the conditional expectation $f(y) = \ex{ U^n \mid Y^n = y}$. %
We note that this formulation of the distortion-rate function does not necessarily describe the optimal performance that is possible using a random spherical code, because the estimation stage is based only on the compressed representation $Y^n$ and does not use any other information about the realization of the codebook. \par
Following the central theme of this paper, our results are described in terms of an AWGN counterpart to the distortion-rate function. Given noise variance $\sigma^2$, define the sequence $\{Z_n\}_{n \in \mathbb{N}}$ by
\begin{align}
Z_n = X_n + \sigma W_n,
\end{align}
where $W_n$ is an independent standard Gaussian noise. The MSE associated with an estimator $f : \reals^n \to \reals^n$ is defined by
\begin{align}
\Mcal_n(\sigma^2, f) & \coloneqq  \ex{ d\left( U^n , f(Z^n) \right)},%\\
%\cM_n(\sigma^2) & \coloneqq \min_{f}  \ex{ \| U^n - f(Z^n) \|^2},
\end{align}
The minimum over $f$ is attained by the conditional expectation $f(z) = \ex{ U^n \mid Z^n  = z}$ and is denoted by $\Mcal_n(\sigma^2) \coloneqq \min_{f} \Mcal_n(\sigma^2,f)$. Stationarity of the sequence $\{(U_n,Z_n)\}$ implies that $n \Mcal_n(\sigma^2)$ is sub-additive in $n$, and thus the following limit is well-defined
\begin{align}
\Mcal(\sigma^2) & \coloneqq \lim_{n \to \infty}  \Mcal_n(\sigma^2).
\end{align}
We refer to $\Mcal(\sigma^2)$ as the minimum MSE function associated with the AWGN model. 
The next result establishes the formal equivalence between the distortion-rate function associated with random spherical coding and $\Mcal(\sigma^2)$. The proof is based on the Gaussian approximation of quantization error in Theorem~\ref{thm:Wpbound} as well as some further properties of the AWGN model.
\begin{thm}
\label{thm:indirect_source_coding} 
Suppose that  $\{(U_n,X_n)\}$ is a random process satisfying Assumption~\ref{assumption:source_coding}. Let $\{f_n\}_{n \in \mathbb{N}}$ be a sequence of estimators $f_n: \reals^n \to \reals^n$ satisfying  $ \| f_n\|_\mathrm{Lip} \le L$ and $\|f_n(0)\|\le \sqrt{n} C$  for all $n$ where $L,C$ are positive constants. Then, for each $R>0$, 
\begin{align}
\lim_{n \to \infty}  \left|  \Dsp_n(R,f_n) - \Mcal_n(\sigma_R^2,f_n)\right| = 0 ,
\end{align}
where $\sigma^2_R = \ex{|X_1|^2} / (2^{2R}-1)$.
Furthermore, there exists a sequence of estimators $\{f_n\}_{n\in\mathbb N}$ such that
\begin{align}
\lim_{n \to \infty} \frac{1}{n} \Dsp(R,f_n)= \lim_{n \to \infty} \frac{1}{n} \Mcal(\sigma^2_R,f_n) = \Mcal(\sigma^2_R).
\end{align}
\end{thm}
\begin{proof}
Set $\gamma = \sqrt{n}\sqrt{\ex{|X_1|^2}}$, $M = \lceil 2^{ n R} \rceil$, and $(\rho,\sigma)$ as in \eqref{eq:rhosig}. Note that $\sigma^2 = \sigma_R^2$. Following the same steps as in the proof of Proposition~\ref{prop:Wp_cont}, we have
\begin{align}
\left| \sqrt{ \Dsp_n(R,f_n) } -\sqrt{ \Mcal_n(\sigma^2_R,f_n)} \right|
& \le \frac{ L \cdot W_2(P_{Y^n},P_{Z^n}) }{\sqrt{n}}.
\end{align}
By Theorem~\ref{thm:Wpbound}, the normalized Wasserstein distance can be upper bound as
\begin{align*}
\frac{W_2(P_{Y^n},P_{Z^n})}{\sqrt{n}} & \le \left(  \ex{ \left| \frac{1}{\sqrt{n}} \|X^n\| -\sqrt{\ex{ |X_1|^2}} \right|^2}\right)^{1/2}  \\
& \qquad +  \frac{\sqrt{2} C_R}{\sqrt{2^{2R}-1}} \frac{ \sqrt{\ex{ |X_1|^2}}}{\sqrt{n}},
\end{align*}
where $C_R$ is a constant that depends only on $R$. The second term in this bound converges to zero at a rate $1/\sqrt{n}$. Combining the inequality $|\sqrt{a}- \sqrt{b}| \le \sqrt{|a-b|}$ with the assumption that $\{X_n\}$ is second order ergodic, one finds that the first term also converges to zero. Putting everything together, we conclude that
\begin{align}
\lim_{n \to \infty} \left| \sqrt{ \Dsp_n(R,f_n) } -\sqrt{ \Mcal_n(\sigma^2_R,f_n)} \right|  = 0. 
\end{align}

To prove that this comparison holds without the square roots, it is sufficient to show that $\Dsp_n(R,f_n)$ and $\Mcal_n(R,f_n)$ are bounded uniformly with respect to $n$. To this end, we can use the triangle inequality and the assumptions on $f_n$ to write:
\begin{align*}
 \|U^n - f_n(Y^n)\| & \le  \|U^n -  f_n(0)\|  + \| f_n(0)\|+\|f(Y^n)\| \\
 & \le \|U^n\| +\sqrt{n} C + L \|Y^n\|.
\end{align*}
Combining this bound with the assumptions on $U^n$ and $Y^n$ establishes that $\Dsp_n(R,f_n)$ is bounded uniformly, and the same approach also works for $\Mcal_n(\sigma^2_R, f_n)$.

To prove the second part, we will show that for each $\eps> 0$, there exists a sequence of  estimators $f_n$ satisfying $\sup_n \|f_n\|_\mathrm{Lip}<\infty$ and  $\limsup_{n \to \infty}  \cM(\sigma^2, f_n) \le \cM(\sigma^2) + \eps$. The existence of the limit in the definition of $M(\sigma^2)$ means that for each $\eps > 0$,  there exists an integer $N$ such that $| \Mcal_n(\sigma^2) - M(\sigma^2) | \le \eps$ for all $n \ge N$. By Lemma~\ref{lem:Lip_approx} in the Appendix, there exists a Lipschitz continuous function $g: \reals^N \to \reals^N$ such that $| \Mcal_n(\sigma^2) - \Mcal_n(\sigma^2 ,g)| \le \eps $. For $n \ge N$, let $f_n : \reals^n \to \reals^n$ be defined by applying $g$ to the first $\lfloor n/N \rfloor$ successive length-$N$ blocks of $Z^n$  and setting any remaining entries to zero. Then, we have  $\|f_n\|_\mathrm{Lip} = \|g\|_\mathrm{Lip}$  and%for all $n$, and % By the stationarity of the sequence $\{(U_n, Z_n)\}$, it follows that 
\begin{align*}
& \Mcal_n(\sigma^2, f_n) \\
& \qquad = \frac{1}{n} \lfloor n/N \rfloor \Mcal_n(\sigma^2, g) + \frac{1}{n}\left (n -  \lfloor n/N \rfloor \right) \ex{ |U_1|^2}.
\end{align*}
Putting everything together, we have $|\Mcal_n(\sigma^2, f_n) - M(\sigma^2)| \le 3 \eps$ for all $n$ large enough. As $\eps$ can be chosen arbitrarily small, the proof is complete. 
\end{proof}

The significance Theorem~\ref{thm:indirect_source_coding} is that it provides a link between the problem of estimation from compressed data, which is often difficult to study directly, and the better-understood problem of estimation in Gaussian noise. We emphasize that the assumptions on the source and data are quite general, particularly in comparison to many of the existing results in the literature. 

Compared to optimal encoding schemes that attain the indirect distortion-rate function $D(R)$, a useful property of random spherical coding is that it can be implemented without any knowledge of the underlying source distribution. Therefore, the coding scheme described in this paper can be employed in typical data acquisition situations where the distribution of the data and the source of interest is learned \emph{after} the data are collected and quantized.

\subsection{Universality of Linear Estimation}
We now consider the performance of linear estimators.
Given a bitrate $R >0$, define the scalar
\begin{align}
\alpha_R =  \left(1 - 2^{-2R} \right) \frac{ \ex{ U_1 X_1}}{ \ex{|X_1|^2}}. \label{eq:alpha_R}
\end{align}
A standard calculation reveals that under the AWGN model, the MSE of the linear estimator $f(y) = \alpha_R y$ is independent of the problem dimension and is given by 
\begin{align}
\Mcal_n(\sigma^2_{R},f) = \frac{1}{n} \ex{ \| U^n - \alpha_R Z^n\|^2}
 = D_G(R),
\end{align}
where we recall that $D_G(R)$ of \eqref{eq:Gaussian_iDRF} is the  distortion-rate function associated with a zero-mean Gaussian source. In view of Theorem~\ref{thm:indirect_source_coding}, this correspondence between the Gaussian distortion-rate function and the MSE of linear estimators in the AWGN model implies an achievable result for random spherical codinig.

\begin{prop}\label{prop:Dsp_lin}
Let $\{(U_n,X_n)\}$ be a process satisfying Assumption~\ref{assumption:source_coding}. For each integer $n$, let $Y^n$ be the output of a bitrate-$R$ random spherical code with input $X^n$ and squared magnitude $n\ex{|X_1|^2}/ (2^{2R}-1)$. Then, 
\begin{align}
% \lim_{n \to \infty} \Dsp_n(R,f_n) = 
\lim_{n \to \infty} \frac{1}{n} \ex{ \| U^n - \alpha_R Y^n\|^2} = D_G(R),
\end{align}
where $\alpha_R$ is given by \eqref{eq:alpha_R}. 
\end{prop}

Applied to the special case of direct source coding $X^n=U^n$, Proposition~\ref{prop:Dsp_lin} recovers the results in \cite{sakrison1968geometric} and \cite{lapidoth1997role}, which showed that squared error distortion of a (properly scaled) random spherical code depends only on the second-order statistics of the source and is equal to the Gaussian distortion-rate function. The contribution of Proposition~\ref{prop:Dsp_lin} is to show that this result carries over naturally to the indirect source coding setting. Moreover, if $\{(U_n, X_n)\}$ are i.i.d.\ zero-mean Gaussian, then we have the equivalence:
\begin{align}\label{eq:equivalence}
    D(R) = \Mcal(\sigma^2_R) = D_G(R).
\end{align}
We note that, in general, 
codebooks approaching the optimal trade-off between bitrate and MSE described by $D(R)$ depend on the joint distribution of $\{(U_n,X_n)\}$. This is because such codebooks essentially encode the sequence obtained by estimating $U^n$ from $X^n$ \cite{1056251,kipnis2021rate}, i.e., estimation precedes encoding in this case. When $U^n$ and $X^n$ are i.i.d.\ and jointly Gaussian, this estimation is obtained by multiplying $X^n$ by $\alpha_R$, and there is essentially no difference if this multiplication is performed pre- or post-encoding. To summarize, for i.i.d. Gaussian and zero mean 
$\{(U_n, X_n)\}$, the equality $D(R) = \Mcal(\sigma^2_R)$ is due to two factors: (1) The optimal estimator is a scalar multiple of the data, and (2) random spherical coding is optimal for encoding Gaussian sources. \par

\subsection{Non-linear Estimation}
Next, we consider the performance of non-linear estimators when the source and the data are non-Gaussian. Suppose that the source and the data are memoryless, that is the pairs $(U_n,X_n)$ are i.i.d.\ from a distribution $P_{U,X}$ with finite second moments. Under this assumption, the indirect distortion-rate function $D(R)$ can be expressed as \cite{DobrushinTsybakov,berger1971rate}
\begin{align}
D(R) = \min_{ I(X; \hat{U}) \le R} \ex{(U - \hat{U})^2},
\end{align}
where the minimum is over all distributions on $(X,\hat{U})$ such that $X \sim P_X$,
%$U \to X \to \hat{U}$ forms a Markov chain, 
and $I(X;\hat{U}) \le R$. Noting that 
\[
\min_{ I(X; \hat{U}) \le R} \ex{(U - \hat{U})^2} = \min_{ I(X; \hat{X}) \le R} 
\ex{d(X,\hat{X})},
\]
where $d(x,\hat{x}) \coloneqq \ex{(U-\hat{U})^2 \mid X = x, \hat{U}=\hat{x}}$, $D(R)$ can be approximated numerically using \cite{Blahut}.\par
In the setting of the AWGN model, the memoryless assumption means that the problem of estimating $U^n$ from  $Z^n$ decouples into $n$ independent estimation problems, and the minimum MSE function is given by
\begin{align}
    \Mcal(\sigma^2) = \ex{ (U - \ex{ U \mid  Z})^2},  \quad Z = X + \sigma W,
\end{align}
where $(U,X) \sim P_{U,X}$ and $W \sim \normal(0,1)$ are independent. This expression can be approximated numerically using standard techniques. 

An interesting special case of the indirect source coding problem occurs when the data is an AWGN corrupted version of the source, that is
\begin{align}
X = U + \eps W', \label{eq:Xawgn}
\end{align}
with $U \sim P_U$ independent of  $W' \sim \normal(0,1)$.  In this case, the Gaussian noise in the data can be combined with the independent Gaussian noise in the AWGN model such that
 \begin{align}
Z = U + \sqrt{ \eps^2 + \sigma^2} W'',
\end{align}
 where $U\sim P_{U}$ independent of $W'' \sim \normal(0,1)$. In Figure \ref{fig:noisy_BPSK}, we provide a comparison of the indirect distortion-rate function $D(R)$ and the upper bound on the distortion obtained using random spherical coding $\Mcal(\sigma_R^2)$ in the setting where $U$ is uniform on $\{-1,1\}$ and $X$ is drawn according to \eqref{eq:Xawgn}. For comparison, we also plot the upper bound $D_G(R)$ corresponding to linear estimation, as well as the asymptotes of all MSE functions as the noise variance $\epsilon$ vanishes.  

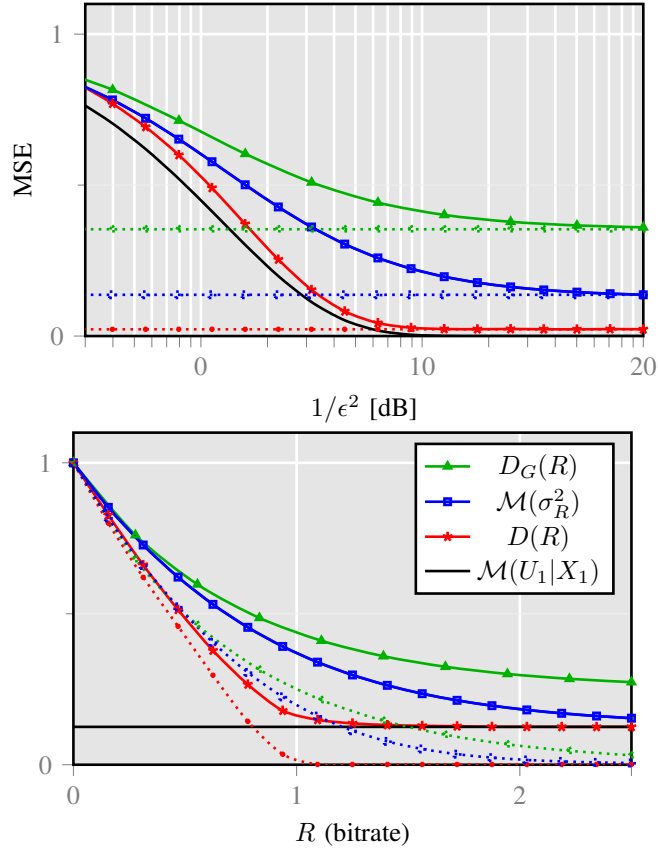
\begin{figure*}
    \centering
    \begin{tikzpicture}
    \begin{semilogxaxis}[
    legend pos=outer north east,
    width=9cm, height=6cm,
    xmin = 0.3, xmax = 100, 
    ymin = 0, ymax=1.1, 
    ytick = {0, 1},
    ylabel = {MSE},
    xlabel = {$1/\epsilon^2$  [dB]},
    yticklabels = {0,1},
    xtick = {0.1,0.2,0.3,0.4,0.5,0.6,0.7,0.8,0.9,1,2,3,4,5,6,7,8,9,10,20,30,40,50,60,70,80,90,100},
    xticklabels={,,,,,,,,,0,,,,,,,,,10,,,,,,,,,20}, 
    line width=1.0pt,
    mark size=1.5pt,
    ]
    \def\R{0.75}
    \addplot[color = blue, solid, smooth, mark = square, mark size=1.2pt, domain=0.2:100] table [x=snr, y=Drsc, col sep=comma] {./PlotData/noisy_BPSK_SD.csv}; %\addlegendentry{$\Mcal(\sigma_R^2,f^1_{\star})$};
    
    \addplot[color = red, smooth, mark = star, mark size=1.9pt] table [x=snr, y=Dind, col sep=comma] {./PlotData/noisy_BPSK_SD.csv}; 
    %\addlegendentry{$D(R,\epsilon^2)$};
    
    \addplot[color = red, dotted, smooth, mark = star, mark size=1.9pt] table [x=snr, y=Dshn, col sep=comma] {./PlotData/noisy_BPSK_SD.csv}; %\addlegendentry{$D(R,0)$};
    
    \addplot[color = black, smooth] table [x=snr, y=mmse, col sep=comma] {./PlotData/noisy_BPSK_SD.csv}; %\addlegendentry{$\mmse(U_1|X_1)$};
    
    \addplot[color = black!30!green, mark = triangle, mark size=1.5pt, smooth, domain=0.2:100, samples = 10] {1 / (1 + x * (1 - 2^(-2 * \R))/(1 + x * 2^(-2*\R) ) ) }; %\addlegendentry{$D_G(R,\epsilon^2)$};
    
    \addplot[color = black!30!green, mark = triangle, mark size=1.5pt, dotted, smooth, domain=0.1:100, samples = 21] {2^(-2*\R)}; %\addlegendentry{$D_G(R,0)$};
    
    \addplot[color = blue, solid, smooth, mark size=1.4pt, domain=0.01:100] table [x=snr, y=Drsc, col sep=comma] {./PlotData/noisy_BPSK_SD.csv};
    
    \addplot[color = blue, dotted, smooth, mark = square, mark size=1.2pt, domain=0.01:100] table [x=snr, y=Dinf, col sep=comma] {./PlotData/noisy_BPSK_SD.csv}; %\addlegendentry{$\Mcal(\sigma_R^2,f^1_{\star})$};
\end{semilogxaxis}
\end{tikzpicture}
\begin{tikzpicture}
    \begin{axis} [
    legend pos = north east,
    width=9cm, height=6cm,
    xmin = 0, xmax = 2.5, 
    ymin = 0, ymax=1.1, 
    ytick = {0, 1},
    xlabel = {$R$ (bitrate)},
    yticklabels = {0,1},
    xtick = {0,1,2,2.5},
    xticklabels={0,1,2}, 
    line width=1.0pt,
    mark size=1.5pt,
    ]
    \def\S{3}
    
    \addplot[color = black!30!green, mark = triangle, mark size=1.5pt, smooth, domain=0:2.5, samples = 10] {1 / (1 + \S * (1 - 2^(-2 * x))/(1 + \S * 2^(-2*x) ) ) }; \addlegendentry{$D_G(R)$};
    
    \addplot[color = blue, solid, smooth, mark = square, mark size=1.2pt, domain=0:2.5] table [x=R, y=Drsc, col sep=comma] {./PlotData/noisy_BPSK_RD.csv}; \addlegendentry{$\Mcal(\sigma_R^2)$};
    
    \addplot[color = red, smooth, mark = star, mark size=1.9pt] table [x=R, y=Dind, col sep=comma] {./PlotData/noisy_BPSK_RD.csv}; 
    \addlegendentry{$D(R)$};

    \addplot[color = black, smooth] table [x=R, y=mmse, col sep=comma] {./PlotData/noisy_BPSK_RD.csv}; \addlegendentry{$\Mcal(U_1|X_1)$};

    \addplot[color = red, dotted, smooth, mark = star, mark size=1.9pt] table [x=R, y=Dshn, col sep=comma] {./PlotData/noisy_BPSK_RD.csv}; %\addlegendentry{$D(R,0)$};

    \addplot[color = black!30!green, mark = triangle, mark size=1.5pt, dotted, smooth, domain=0:2.5, samples = 10] {2^(-2*x)}; %\addlegendentry{$D_G(R)$};
    
    \addplot[color = blue, solid, smooth, mark size=1.4pt, domain=0:2.5] table [x=R, y=Drsc, col sep=comma] {./PlotData/noisy_BPSK_RD.csv};
    
    \addplot[color = blue, dotted, smooth, mark = square, mark size=1.2pt, domain=0:2.5] table [x=R, y=Dinf, col sep=comma] {./PlotData/noisy_BPSK_RD.csv}; %\addlegendentry{$\Mcal(\sigma_R^2,f^1_{\star})$};
\end{axis}
\end{tikzpicture}
    \caption{Mean square error (MSE) in estimating an i.i.d. signal equiprobable on $\{-1,1\}$ from a bitrate-$R$ encoding of its AWGN-corrupted version. Left Panel: MSE versus noise variance $\epsilon^2$ with a fixed encoding bitrate $R = 1$. Right Panel: MSE versus bitrate $R$ with noise variance $\epsilon^2 = 1/3$. $\Mcal(\sigma_R^2)$ is achievable using a random spherical code followed by a scalar Bayes estimator. $D_G(R)$ is achievable using random spherical coding followed by a scalar linear estimator. $D(R)$ is the indirect distortion-rate function corresponding to the optimal encoding scheme. The dashed lines indicate the asymptotic MSE as $\epsilon \to 0$. Also shown is $\Mcal(U_1|X_1)$, which is the minimal MSE in estimating the signal from its corrupted version corresponding to the limit $R\to \infty$. 
    }
    \label{fig:noisy_BPSK}
\end{figure*}

\section{Proof of Main Result}
\label{sec:proof:coupling}
\label{sec:main_proof}

% \subsection{
% Proof of Theorem~\ref{thm:coupling}
% \label{sec:proof:coupling}
% }
The proof of Theorem~\ref{thm:coupling} requires several lemmas.
\begin{lem}\label{lem:sphere_decomp}  Suppose that $V$ is distributed uniformly on the unit sphere in $\reals^n$ with $n \ge 2$. For any $x \in \reals^n\backslash\{0\}$, the distribution on $V$ can be decomposed as
\begin{align}
V  = G \frac{x}{\|x\|} + \sqrt{ 1- G^2}H 
\end{align}
where $G = \langle x,V \rangle /\|x\|$ is a random variable supported on $[-1,1]$ with complementary cumulative distribution function 
%\nr{GR: Lets define $Q_n$ earlier when give distribution of $\alpha^*$. For consistency we maybe give expressions in term of an angle $\alpha$ instead if $G = \cos \alpha$. You mean like
$\pr{ G \ge g} = Q_n(g)$ of \eqref{eq:Qn},  %\frac{\Gamma(\frac{n}{2})}{ \sqrt{\pi}\Gamma(\frac{n-1}{2})}
% \kappa_n \int_g^1  (1 - t^2)^{\frac{n-3}{2}} \, dt,
% \label{eq:Qn}
% \\
% & 
% \kappa_n \coloneqq \frac{\Gamma\left(\frac{n}{2}\right)} {\sqrt{\pi}\Gamma\left(\frac{n-1}{2}\right)}, \nonumber
% \end{align}
and $H$ is an independent random vector distributed uniformly on the set $\{ h \in \reals^n \, : \, \|h\| =1 \text{ and } \langle x, h \rangle =0\}$.
\end{lem}
\begin{proof}
By the orthogonal invariance of the distribution on $V$, we may assume without loss of generality that $x$ is a unit vector of the form $x = (1, 0, \dots 0)$. Then, $G = V_1$ and $H = (0, V_2, \dots, V_n)/\sqrt{ \sum_{i=2}^n V_i^2}$. The joint distribution of $(G,H)$ follows from the joint distribution on the entries of a random spherical vector~\cite[Eq.~(3)]{stam:1982}.
\end{proof}

\begin{lem}\label{lem:PYx}
For $n \ge 2$, let $Y$ be the output of an $(n,M)$-random spherical code with input $x \in \reals^n\backslash \{0\}$ and magnitude $\rho$. The distribution of $Y$ can be decomposed as 
\begin{align}
Y =  \rho  \left(  \frac{x}{\|x\|} \cos(\alpha^*) + \sqrt{ 1- S^2}H \right)
\end{align}
where $\cos(\alpha^*)$ has the cumulative distribution function \eqref{eq:alpha_CDF},
% \begin{align}
% \pr{S \le s} = (1 - Q_n(s))^M,  \label{eq:distF}
% \end{align}
% %distributed according to \eqref{eq:distF}, 
% where $Q_n(t)$ is given in \eqref{eq:Qn} 
and $H$ is an independent random vector  distributed uniformly on the set $\{ h \in\reals^n \, : \, \|h\| = 1, \text{ and }  \langle x, h \rangle = 0\}$.
\end{lem}
\begin{proof}
For each code word $C(i)$ we apply the decomposition in Lemma~\ref{lem:sphere_decomp} to obtain
\[
C(i)  =  G(i)  + \sqrt{1 - G(i)^2} H(i) ,
\]
where $G(i) = \langle x, C(i) \rangle / \|x\| \|C(i)\|$ is the cosine similarity of the $i$-th codeword. Recall that the index $i^*$ corresponds to the code word that maximizes the cosine similarity $
\cos(\alpha^*) = G(i^*) \coloneqq \max \left\{ G(1), \dots G(M) \right\}$. 
Therefore, the distribution of $S$ follows from the fact that $G(1), \dots G(M)$ are i.i.d.\ with complementary cumulative distribution function given by \eqref{eq:Qn}. Furthermore, because $i^*$ depends only on the terms $G(1), \dots G(M)$, it follows from Lemma~\ref{lem:sphere_decomp} that $H \coloneqq H(i^*)$ is independent of $\alpha^*$ and uniform on the subset of the unit sphere that is orthogonal to $x$. Noting that $Y= \rho C(i^*)$ completes the proof. 
\end{proof}

\begin{lem}\label{lem:PZx}
Suppose that $W$ is a standard Gaussian vector on $\reals^n$ with $n\ge 2$. For any  $x \in \reals^n\backslash\{0\}$ the random vector $Z= x + \sigma W $ can be decomposed as
\begin{align}
Z =  x +  \sigma A \frac{x}{\|x\|} + \sigma B H 
\end{align}
where $(A,B,H)$ are independent, $A \sim \normal(0,1)$ has a standard Gaussian distribution, $B \sim \chi_{n-1}$ has a chi distribution with $n-1$ degrees of freedom, and $H$ is distributed uniformly on the set $\{ h \in\reals^n \, : \, \|h\| = 1, \text{ and }  \langle x, h \rangle = 0\}$.
\end{lem}

\begin{proof}
By the orthogonal invariance of the Gaussian distribution on $W$, we may assume without loss of generality that $x$ is a unit vector of the form $x = (1,0, \dots, 0)$. Letting $A =W_1$, $B = \sqrt{ \sum_{i=2}^n W_i^2}$, and $H = (0, W_2,\dots, W_n)/B$ yields $W =  A x/ \|x\| + \ B H $. By construction, $A$ is a standard Gaussian variable that is independent of $W_2, \dots, W_n$. The distribution of $(B,H)$ follows from the fact that $B (H_2, \dots, H_n)$ is the polar decomposition of the $(n-1)$-dimensional standard Gaussian vector  $(W_2, \dots W_n)$.
% Let $\Psi = [\psi_1, \dots, \psi_n] \in \reals^{n \times n}$ be an orthogonal basis for $\reals^n$ whose first column is given by $\psi_1 = x/\|x\|$. By the orthogonal invariance of the standard Gaussian, we can write $W = \sum_{i=1}^n \psi_i  \tilde{W}_i $ where $\tilde{W} = (\tilde{W}_1, \dots, \tilde{W}_n)$ is a standard Gaussian vector. Letting $U = \tilde{W}_1$, $V = \| \sum_{i=2}^n \tilde{W}_i \psi_i \|$, and $H =\sum_{i=2}^n \tilde{W}_i \psi_i / V $ we then have
% \begin{align*}
%  Z = x + \sigma W = x + \sigma A \frac{x}{\|x\|} + \sigma B H.
% \end{align*}
% By construction $A$ is standard Gaussian that is independent of $(B,H)$ and $H$ is a unit vector that is orthogonal to $x$.  The distribution of $B$ follows from noting that  $B^2 = \sum_{i=1}^n \tilde{W}_i^2$ is the sum of $(n-1)$ independent squared Guassians, and thus has  chi-squared distribution with $(n-1)$ degrees of freedom. Finally, the distribution of $H$ and the fact that it is independent of $B$ follows form the orthogonal invariance of the standard Gaussian distribution. 
\end{proof}

% \subsection{Proof of Theorem~\ref{thm:prob_bound}}\label{proof:thm:prob_bound}

Using the characterizations of $Y$ and $Z$ given in Lemma~\ref{lem:PYx} and Lemma~\ref{lem:PZx}, respectively, we see that for every $x \in \reals^{n} \backslash \{0\}$, there exists a coupling on $(Y,Z)$ such that %\nr{GR: Fix notation with $\alpha^*$ maybe.} \nb{AK: Done.}
\begin{align*}
Y &= x + \left( \rho \cos(\alpha^*)  - \|x\| \right) \frac{x}{\|x\|} +  \rho \sin(\alpha^*) H,\\
Z &=  x +  \sigma A \frac{x}{\|x\|} + \sigma B H,
\end{align*}
where $(A,B,H,\alpha^*)$ are independent. By the orthogonality of $x$ and $H$, the magnitude of the difference between $Y$ and $Z$ depends only on the tuple $(A,B,\alpha^*)$ and is given by \eqref{eq:coupling}.

\section{Conclusions
\label{sec:conclusions}}
We considered the problem of estimating an underlying signal or parameter from the lossy compressed version of another high dimensional signal. For compression codes defined by a random spherical code of bitrate $R$, we showed that the distribution of the output codeword is close in Wasserstein distance to the conditional distribution of the output of an AWGN channel with SNR $2^{2R}-1$. This equivalence between the noise associated with lossy compression and an AWGN channel allows us to adapt existing techniques for inference from AWGN-corrupted measurements to estimate the underlying signal from the compressed measurements, as well as to characterize their asymptotic performance. \par
We demonstrated the usefulness of this equivalence by deriving novel expressions for the achievable risk in various source coding and parameter estimation settings. These include bitrate-constrained sparse parameter estimation using soft thresholding, bitrate-constrained parameter estimation in high-dimensional linear models, and indirect source coding with linear and non-linear decoders. In each of these settings, our results yielded achievable MSE and provided the equivalent noise level required to tune the estimator to attain this MSE.\par
We believe that the characterization of lossy compression error developed in this paper can be useful in numerous important cases aside from the ones we explored. Examples of such cases include hypothesis testing based on compressed data, signal estimation in distributed lossy compression settings, and the study of convergence rates and accuracies of first-order optimization procedures employing gradient compression. 

\section*{Acknowledgments}
 The authors wish to thank Cynthia Rush, David Donoho, and Robert Gray for helpful discussions and comments. We are also grateful for the two anonymous reviewers that encouraged us to derive a much stronger version of the main results than we initially reported. The work of G. Reeves was supported in part by funding from the NSF under Grant No.~1750362. The work of A. Kipnis was supported in part by funding from the Koret Foundation and the NSF under Grant No.~DMS-1816114.

\ifdefined\PROOFS

\appendices
%\tableofcontents

{
\section{
\label{sec:Proof2_3}
%Distribution of relative error
Proofs of Theorem~\ref{thm:DeltaR} and \ref{thm:Delta_beta}
}

The proof of Theorems~\ref{thm:DeltaR} and \ref{thm:Delta_beta} require a characterization of the moments of the random variable
\begin{align}
    \Delta \coloneqq \frac{1}{\sigma} \sqrt{  \left(\rho \cos(\alpha^*) - \gamma - \sigma A \right)^2 + \left( \rho \sin(\alpha^*) - \sigma B \right)^2} %+ \sqrt{A^2 + (B - \sqrt{n})^2}
\end{align}
where $(\rho, \sigma, \gamma)$ are deterministic parameters and  $(\alpha^*, A, B)$ are independent random variables whose distributions are described in Theorem~\ref{thm:coupling}. 

Given $\alpha \in (0, \pi/2)$ let 
\begin{align}
    \rho = \gamma \sec \alpha, \qquad \sigma =  \frac{ \gamma \tan \alpha}{ \sqrt{n}}.
\end{align}
Evaluating $\Delta$ with these values and then using the triangle inequality as well as basic trigonometric identities leads to
\begin{align}
\label{eq:proof:main:Delta}
    \Delta &\le \Delta_1 + \Delta_2,  
\end{align}
where 
\begin{align}
    \Delta_1 & \coloneqq 2 \sqrt{n}\csc(\alpha)  \left|\sin\left( \frac{ \alpha^* - \alpha}{2} \right) \right|, \\
    \Delta_2 & \coloneqq
    \sqrt{ A^2 + (B - \sqrt{n})^2}.
\end{align}
The term $\Delta_2$ is sub-Gaussian with mean and variance parameter independent of $n$. An estimate for its sub-Gaussian constant is provided in  Lemma~\ref{lem:Delta2}.  For the term, $\Delta_1$, we use the following result, which is proved in Appendix~\ref{proof:lem:alpha_bound}.

\begin{lem}\label{lem:alpha_bound}
Suppose that $M = M_{\alpha}(n) \coloneqq \sqrt{n}(\csc \alpha)^{n-1}$ for $\alpha \in (0, \pi/2)$. Then, for $p \ge 1$, 
\begin{align}
\label{eq:lem:alpha_bound}
   & \ex{ \left|\sin\left( \frac{ \alpha^* - \alpha}{2} \right) \right|^p}^{1/p}  \\
   & \qquad \le C \frac{ \tan(\alpha) [ \log(n \wedge \sec(\alpha) ) + p ]}{n} + 2^{-M/p}  \end{align}
where $C$ is a numerical constant. 
\end{lem}
\begin{proof}[Proof of Theorem~\ref{thm:DeltaR}] 
In view of \eqref{eq:proof:main:Delta} and the fact that $\Delta_2$ is sub-Gaussian with a constant that does not depend on on $n$ all that remains is to establish the desired upper bound on the moments of  $\Delta_1$. Note that if $p \ge n$ then this term is bounded almost surely according to 
\[
\Delta_1 \leq 2 \csc(\alpha) \sqrt{n} \leq 2 \csc(\alpha) \sqrt{p}, 
\]
where $2 \csc (\alpha)$ depends only on $R$. The remainder of the proof focuses on the case $1 \le p \le n$.

Recall that the specification in \eqref{eq:rhosig} corresponds to the choice $\alpha = \arcsin(2^{-R})$. 
% If $n \leq p$, then 
% \[
% \Delta_1 \leq 2 \csc(\alpha) \sqrt{n} \leq 2 \csc(\alpha) \sqrt{p}. 
% \]
% It is left to consider the case $n > p$.
Define 
\begin{align}
    N_R = \min\left\{ n \in \mathbb{N} \, : \, 2^{R +1} \le \sqrt{n} \le 2^{\frac{n+1}{2} R} \right\}
\end{align}
For $n \ge N_R$ it can be verified that there exists a unique value $\alpha_n \in [\alpha, \pi/2)$ such that $M = \sqrt{n} \csc(\alpha_n)^{(n-1)}$  and $ \sin (\alpha_n) \le \sqrt{ \sin (\alpha)}$. 
% \begin{align}
%  \sin (\alpha)\le \sin (\alpha_n) \le \sqrt{ \sin (\alpha)}
%\end{align}
Noting the the sin function is Lipschitz and non-decreasing on $[0, \pi/2]$ we can write
\begin{align}
\left|\sin\left( \frac{ \alpha^* - \alpha}{2} \right) \right| \le \left|\sin\left( \frac{ \alpha^* - \alpha_n}{2} \right) \right| + \sin\left( \frac{ \alpha_n - \alpha}{2} \right). \label{eq:sin_alpha_R_decomp}
\end{align}
The second term in \eqref{eq:sin_alpha_R_decomp} is deterministic and satisfies 
\begin{align}
   \sin\left( \frac{ \alpha_n - \alpha}{2} \right) & = \cos\left( \frac{ \alpha_n - \alpha}{2} \right)  \frac{ \sin(\alpha_n) - \sin(\alpha)}{ \cos(\alpha_n)  + \cos(\alpha)}\\
   & \le \frac{ \tan(\alpha_n)}{2} \left(1 - \frac{ \sin(\alpha)}{ \sin(\alpha_n)} \right)\\
& \le \frac{ \tan(\alpha_n)}{2} \log \left( \frac{ \sin(\alpha_n)}{ \sin(\alpha)} \right)\\
& =  \frac{ \tan(\alpha_n)}{2}  \left[ \frac{\frac{1}{2} \log n - \log M }{n-1}  + R \log(2) \right]\\
& \le      \frac{ \tan(\alpha_n) \log n}{2 n} . 
\end{align}
To bound the moments of the first term in \eqref{eq:sin_alpha_R_decomp} we can  use  Lemma~\ref{lem:alpha_bound}. Finally, recalling that $\sin(\alpha_n) \le \sqrt{ \sin \alpha}$, it follows that 
\begin{align}
    \tan(\alpha_n) &= \frac{1}{ \sqrt{ \sin(\alpha_n)^{-2} - 1} } \\
    &\le \frac{1}{ \sqrt{ \sin(\alpha)^{-1} - 1} }.
\end{align}
In view of \eqref{eq:lem:alpha_bound} and $n >p$, it follows that
\begin{align}
    \ex{ \Delta_1^p}^{1/p} & \le C_R \left(  \frac{  [ \log(n) + p ]}{\sqrt{n} } + \sqrt{n} 2^{-M/p}  \right)  \\
    & \le C_R \left( 1 + \sqrt{p} + \sqrt{n} 2^{-M/p} \right) .
\end{align}
%\nbr{From here I think everything works fine. The main trick in the above is the double bound $\sin \alpha \le \sin \alpha_n \le\sqrt{ \sin \alpha}$ which deals with the $\tan(\alpha_n)$ term. }
Finally, for the term $\sqrt{n}2^{-M/p}$,
recalling that $M_n = \lceil 2^{nR} \rceil$, we see that for $p \le \sqrt{n}$, 
\begin{align}
 \sqrt{n}2^{-M/p} & \le \sqrt{n} 2^{-2^{nR}/p} \le \sqrt{n} 2^{-2^{nR}/n} \le C''_R
\end{align}
for some positive constant $C''_R$. This completes the proof of Theorem~\ref{thm:DeltaR}. 
\end{proof}

\begin{proof}[Proof of Theorem~\ref{thm:Delta_beta}] The specification given \eqref{eq:rhosig_tilde} corresponds to the choice 
\begin{align}
\alpha  = \arcsin \left( M^{1/(n-1)} \right)
\end{align}
Under the assumption $M \ge M_\beta(n)$, there exists $\tilde{\alpha} \in (\alpha,\beta]$ such that $M = \sqrt{n} (\csc(\tilde{\alpha}))^{n-1}$.  Using the same  approach as in the proof Theorem~\ref{thm:DeltaR} leads to
\begin{align}
\left|\sin\left( \frac{ \alpha^* - \alpha}{2} \right) \right| \le \left|\sin\left( \frac{ \alpha^* - \tilde{\alpha}}{2} \right) \right| + \frac{ \tan(\alpha_n) \log n}{2 n}
%\sin\left( \frac{ \tilde{\alpha} - \alpha}{2} \right). 
\label{eq:sin_alpha_beta_decomp}.
\end{align}
By Lemma~\ref{lem:alpha_bound} it follows that 
\begin{align}
    \ex{ \Delta_1^p}^{1/p} & \le C'  \frac{ \sec(\tilde{\alpha})  [ \log(n) + p ]}{\sqrt{n} } + 2 \sqrt{n} \csc(\tilde{\alpha}) 2^{-M/p} 
\end{align}
where $C'$ is a numerical constant. 
Since the secant function is non-decreasing on $[0, \pi/2]$ we have $\sec(\alpha_n) \le \sec(\beta)$, and so the first term is bounded from above by $C'_\beta p$ for some number $C'_\beta$. The second term satisfies
\begin{align*}
 \sqrt{n} \csc(\tilde{\alpha}) 2^{-M/p} & =   \sqrt{n} \csc (\tilde{\alpha}) \exp\left\{\frac{ \sqrt{n}}{p} (\csc \tilde{\alpha})^{n-1} \log 2 \right\}  \\
& \le     \sqrt{n} \csc (\tilde{\alpha}) \exp\left\{ - \frac{ \sqrt{n}}{p} \csc (\tilde{ \alpha}) \log 2 \right\}  \\
& \le   \frac{p}{ e \log 2}.
\end{align*}
\end{proof}

% For Theorem~\ref{thm:Delta_beta}, we take the maximum over $M \ge M_\beta(n)$. Equivalently, 
% \begin{align}
% &\max_{\alpha \in [0, \beta]}  \sqrt{n}  \csc(\alpha) \ex{ \left|\sin\left( \frac{ \alpha^* - \alpha}{2} \right) \right|^p}^{1/p}\\
%   & \le \max_{\alpha \in [0, \beta]}  \min\left\{  \frac{ \sec(\alpha) [ \log(n\wedge \sec(\alpha)) + p ]}{\sqrt{n}}  + \sqrt{n} \csc (\alpha) 2^{-M_{\alpha}(n)/p} , \sqrt{n} \csc \alpha \right\}\\
%         & \le     \frac{\sec(\beta) [ \log(n\wedge \sec(\beta)) + p ]}{\sqrt{n}} + \max_{\alpha \in [0, \beta]} \sqrt{n}  \csc (\alpha) 2^{-M_{\alpha}(n)/p} 
% \end{align}
% The second term satisfies
% \begin{align}
%   \max_{\alpha \in [0, \beta]} \sqrt{n}  \csc (\alpha) 2^{-M_{\alpha}(n)/p} 
%   & =   p \max_{\alpha \in [0, \beta]}  \frac{ \sqrt{n}}{p} \csc (\alpha) 2^{-\sqrt{n}(\csc \alpha)^{n-1}/p} \\
% & =   p \max_{\alpha \in [0, \beta]}  \frac{ \sqrt{n}}{p} \csc (\alpha) \exp\left\{\frac{ \sqrt{n}}{p} (\csc \alpha)^{n-1} \log 2 \right\}  \\
% & \le  p \max_{\alpha \in [0, \beta]}  \frac{ \sqrt{n}}{p} \csc (\alpha) \exp\left\{ - \frac{ \sqrt{n}}{p} \csc ( \alpha) \log 2 \right\}  \\
% & \le  p \max_{x \ge 0}   \frac{ x}{ \log 2} \exp\left\{ - x   \right\}\\
% & =  \frac{p}{ e \log 2}
%\end{align}
%\end{proof}
}

\iffalse
\section{
\label{app:proof:bound_R_alpha}
Proof of Lemma~\ref{lem:bound_R_alpha}
}
\nr{AK: Also need to modify here: use $\sqrt{n}$ instead of $\mu_{n-1}$}
{
For $\alpha, \beta \in (0, \pi/2)$, 
\begin{align}
    |\sin (\alpha/2 - \beta/2)| & \le \tan(\beta) \left| \frac{ \sin (\alpha)}{ \sin(\beta)} - 1 \right|\\
    & \le \tan(\beta) \begin{dcases} \frac{ 1}{n-1} \left[ \left(\frac{\sin (\alpha)}{\sin (\beta)} \right)^{n-1}  -1 \right], & \alpha \ge \beta\\
    \frac{ 1}{n-1} \log \left[ \left(\frac{\sin (\beta)}{\sin (\alpha)} \right)^{n-1}  \right],  & \alpha \le \beta
    \end{dcases} 
\end{align}
Evaluating with $\sin \beta = 2^{-2R}$ and $\sin \alpha = (\sqrt{2 \pi \mu_n}/M)^{1/(n-1)}$ leads to 
\begin{align}
    |\sin (\alpha/2 - \beta/2)| \le \frac{\tan(\beta)}{n-1}  \begin{dcases}   T   -1 , & T \ge 1\\
  \log \frac{1}{T} ,  & T \le 1
    \end{dcases} 
\end{align}
where
\begin{align}
T = \frac{ \sqrt{ 2 \pi \mu_n}}{M} 2^{(n-1) R}
\end{align}
In particular, if $M = \lceil 2^{n R} \rceil$ then $2^{nR} \le M < 2^{nR} + 1 \le 2 \cdot 2^{n R}$ and so 
\begin{align}
 \frac{\sqrt{ 2 \pi \mu_n}}{2}  2^{- R} <  T \le  \sqrt{ 2 \pi \mu_n} 2^{- R}
\end{align}
Thus
\begin{align}
    |\sin (\alpha/2 - \beta/2)| & \le \frac{\tan(\beta)}{n-1}  \left[  \left(\sqrt{ 2 \pi \mu_n} 2^{- R} - 1\right)_+ + \log_+ \left( \frac{ 2\cdot 2^R}{ \sqrt{ 2 \pi \mu_n}} \right)   \right] \\
   & \le \frac{\tan(\beta)}{n-1}  \left[  \sqrt{ 2 \pi n}  +  R \log 2 \right]
\end{align}
}
\QEDA
\fi

\section{Proof of Lemma~\ref{lem:alpha_bound}}
\label{proof:lem:alpha_bound}

Let integers $n,M$ and $\alpha \in (0, \pi/2)$ be such that $M = \sqrt{n} (\csc \alpha)^{n-1}$.  Recall that the goal is to bound the absolute moments of the random variable $\sin( (\alpha^* -\alpha)/2)$ where $\alpha^*$ is drawn according to \eqref{eq:alpha_CDF}. We consider two cases. First, on the event $\alpha^* > \pi/2$, we can use the trivial upper bound of one. Note that $Q_n(0) = 1/2$ and so the probability of this event is $2^{-M}$. Alternatively, on the event $\alpha^* \le \pi/2$,  we use basic trigonometric identities to write  
\begin{align}
 \sin\left( \frac{\alpha^* - \alpha}{2} \right)  & = \cos\left( \frac{ \alpha^* - \alpha}{2} \right) \frac{ (\sin (\alpha^*)  - \sin (\alpha)) }{( \cos (\alpha^*) + \cos (\alpha) )}.
\end{align}
Noting that $\cos (\alpha^*) \ge 0$ leads to the  upper bound $|\sin( (\alpha^* -\alpha)/2)| \le \tan(\alpha) \, \xi$ where  
\begin{align}
    \xi \coloneqq \left| \frac{ \sin (\alpha^*)}{\sin (\alpha)} - 1 \right|.
\end{align}
Combining these two cases yields
\begin{align}
    \ex{  \left| \sin\left( \frac{\alpha^* - \alpha}{2} \right) \right|^p}^{1/p} % & = (1 - 2^{-M} ) \ex{ \left| \sin\left( \frac{\alpha^* - \alpha}{2} \right) \right|^p\one_{\alpha^* \le \pi/2}}  +  2^{-M} \ex{ \left| \sin\left( \frac{\alpha^* - \alpha}{2} \right) \right|^p\one_{\alpha^* > \pi/2}} \\
     & \le  \tan(\alpha)  \ex{  \xi^p}^{1/p}  +  2^{-M/p}. 
     \label{eq:sin_alpha_bound}
\end{align}

The remaining step in the proof is to provide an upper bound moments of $\xi$. We need the following bounds on the function $Q_n(s)$ defined in \eqref{eq:Qn}.
%Let \nbr{GR: Again we could refer back to where this is first mentioned. For this proof working with $\mu_n$ is more useful than $\kappa_n$, so we should us only one of the parameterizations. Then, we can remove the following definitions which are redundant.}
%\begin{align}
%Q_n(s) &\coloneqq\frac{\mu_{n-1}}  {\sqrt{2\pi}} \int_s^1  (1 - t^2)^{(n-3)/2}\, dt,
%\end{align}
%so that 
%\begin{align}
% \pr{\cos \alpha  \le s } = \left(1 - Q_n(s) \right)^{M}, \quad s \in [-1,1]. 
%\end{align}
\begin{lem}
For integer $n\ge 2$ and $s \in [0,1]$, 
\begin{align}
\frac{  (1 - s^2)^{(n-1)/2}}{\sqrt{2\pi} \mu_n}    \le   Q_n(s)     \le  \frac{  (1 - s^2)^{(n-1)/2}  }{\sqrt{2\pi} \mu_n \max(s,n^{-1})  }, \label{eq:Qbnd}
\end{align}
where
\begin{align}
    \mu_n \coloneqq \frac{\sqrt{2}\Gamma(\frac{n+1}{2})}{ \Gamma(\frac{n}{2})} \label{eq:mu_n}
\end{align}
is the mean of the chi distribution with $n$ degrees of freedom. 
\end{lem}
\begin{proof}
Making the change of variables $u = (1 -t^2)$ and and using the relation $\mu_{n-1} \mu_n = n-1$ leads to
\begin{align*}
Q_n(s) & = \frac{1}  {\sqrt{2\pi} \mu_n}   \int_0^{1-s^2}  \frac{ (n-1) u^{(n-3)/2} }{ 2 \sqrt{1 - u}} \, du.
\end{align*}
Noting that the denominator in the integral satisfies $ s \le  \sqrt{ 1 - u} \le 1 $ and then integrating gives the double inequality:
  \begin{align}
\frac{  (1 - s^2)^{(n-1)/2}}{\sqrt{2\pi} \mu_n}    \le     Q_n(s)     \le  \frac{  (1 - s^2)^{(n-1)/2}  }{\sqrt{2\pi} \mu_n s }. \label{eq:Q_ub_a}
\end{align}

To prove an upper bound with $s$ replaced by $n^{-1}$ in the denominator, observe that for $n \ge 3$, 
\begin{align*}
 Q_n(s) &= \frac{(n-1) }{\sqrt{2\pi} \mu_{n} } \int_s^1 (1 - t^2)^{(n-3)/2} \, dt\\
& \le \frac{(n-1)}{\sqrt{2\pi} \mu_n}(1 -s^2)^{(n -3)/2} (1- s)\\
& = \frac{ (n-1)(1 -s^2)^{(n -1)/2} }{\sqrt{ 2\pi } \mu_n  (1 + s)}. 
\end{align*}
Meanwhile, for  $n=2$, direct calculation reveals
\begin{align*}
Q_2(s) &= \frac{1}{\sqrt{ 2 \pi} \mu_2 } \left( \frac{\pi}{2} -   \arcsin(s) \right)\\
& \le \frac{2 (1 - s^2)^{1/2} }{\sqrt{ 2 \pi} \mu_2 }.
\end{align*}
Combining these upper bounds completes the proof.
\end{proof}

\begin{lem}\label{lem:CDF_S_bounds}For integers $n,M\ge 2$ and $u \in [0,1]$,
\begin{align}
  \pr{\sin (\alpha^*)  \le u }
& \le  \frac{   u^{n-1}  M}{ \sqrt{ 2 \pi} \mu_n \max( \sqrt{1- u^2}, n^{-1}) } \label{eq:sin_a_bound1}\\
  \pr{ \sin (\alpha^*) \ge  u } 
& \le  \exp\left(  - \frac{   u^{n-1} M}{ \sqrt{ 2 \pi} \mu_n }  \right).\label{eq:sin_a_bound2}
\end{align}
\end{lem}
\begin{proof}
%These bounds hold trivially for $u > 1$ so we will focus on the case $u \in [0,1]$.
Let $s = \sqrt{1-u^2}$. Expressing the event $\{ \sin (\alpha^*) < u\}$ in terms of $\cos (\alpha^*)$, we can write
\begin{align*}
    \pr{ \sin (\alpha^*) \le u} &= \pr{ \cos (\alpha^*) \ge s} + \pr{ \cos (\alpha^*) \le - s}\\
    & =1 - (1 - Q_n(s))^M + (1 - Q_n(-s))^M\\
    & =1 - (1 - Q_n(s))^M +  Q_n(s)^M\\
    & \le M Q_n(s).
\end{align*}
Here, the third step follows because $Q_n(-s)  = 1- Q_n(s)$ and the last step is due to the basic inequality $1 - (1- q)^p -q^p \le p q$ for $q \in [0,1]$ and $p \ge 2$. Combining this inequality with the upper bound on $Q_n(s)$ in   \eqref{eq:Qbnd} gives \eqref{eq:sin_a_bound1}. 
For the complementary event, we have
\begin{align*}
    \pr{ \sin (\alpha^*) \ge u} & = (1 - Q_n(s))^M - (1 - Q_n(-s))^M\\
    & \le  \exp( - M Q_n(s)),
\end{align*}
where we have used the basic inequality $1 - q \le e^{-q}$ for all $q \in [0,1]$. Combining with the lower bound on $Q_n(s)$ in \eqref{eq:Qbnd} gives \eqref{eq:sin_a_bound2}.
\end{proof}

To bound the distribution of $\xi$ we consider two cases. First, conditional on the event $\sin (\alpha^*) \le \sin (\alpha)$, we can write 
\begin{align}
    \xi & = 1 - \frac{\sin (\alpha^*)}{\sin (\alpha)}\le \log \frac{ \sin (\alpha)}{\sin (\alpha^*)},
\end{align}
where we used the basic inequality $1 - 1/x \le  \log x$ for $x > 0$. 
In view of \eqref{eq:sin_a_bound1} and the assumption $M = \sqrt{n} (\csc \alpha)^{n-1}$, it follows that for $t \ge 0$,
\begin{align*}
\pr{ \xi \ge t  , \sin(\alpha^*) < \sin (\alpha)}  & \le \pr{ \sin(\alpha^*) \le  \sin( \alpha ) e^{-t} }  \\
& \le \frac{ \sqrt{n}  \exp\left(-(n-1)t\right)}{ \sqrt{2 \pi} \mu_n \max( \cos(\alpha) , n^{-1} )}\\
& \le  \min(\csc(\alpha),n)  \exp\left(-(n-1)t\right),\label{eq:xi_bound1}
\end{align*}
where the last step uses the lower bound $\mu_n \ge \sqrt{n-1/2}$.

%where $C_{n,\beta} \coloneqq \sqrt{n} \min( \sec(\alpha), n) /(\sqrt{2 \pi} \mu_n)$
% \begin{align}
%     \lambda_{n,M}(\beta) \coloneqq \frac{  M (\sin (\beta))^{n-1}}{ \sqrt{ 2 \pi} \mu_n }.
% \end{align}

Next, conditional on the event $\sin (\alpha^*) \ge \sin (\alpha)$, we can write
\begin{align}
    \xi & =  \frac{\sin( \alpha^*)}{\sin (\alpha)} -1 \le \frac{ 1}{n-1} \left[ \left(\frac{\sin (\alpha^*)}{\sin (\alpha)} \right)^{n-1}  -1 \right],
\end{align}
where we used the basic inequality $x -1 \le (x^p -1)/p$ for all $x \ge 0$ and $p \ge 1$. In view of \eqref{eq:sin_a_bound2} and the assumption $M = \sqrt{n} (\csc \alpha)^{n-1}$, it then follows that for $t \ge 0$, 
\begin{align}
& \pr{ \xi \ge t  , \sin (\alpha^*) \ge \sin (\beta)} \nonumber \\
& \qquad \le \pr{ \sin(\alpha^*) \ge  (\sin (\beta) ) \left[1+ (n-1) t \right]^{\frac{1}{n-1}}  } \nonumber \\
& \qquad  \le \exp\left(  -  \frac{ \sqrt{n} [ 1 + (n-1) t] }{ \sqrt{ 2 \pi} \mu_n}    \right) \nonumber \\
& \qquad \le \exp\left(  -  \frac{(n-1) t }{ \sqrt{ 2 \pi}} \right),\label{eq:xi_bound2}
\end{align}
where the last step uses the upper bound $\mu_n \le \sqrt{n}$.

With \eqref{eq:xi_bound1} and \eqref{eq:xi_bound2} in hand, we can upper bound the moments of $\xi$ according to 
\begin{align*}
    \ex{ \xi^p}  & = \int_0^\infty \pr{ \xi \ge t^{1/p}} \, d t\\
    & \le \frac{\log( \min(\sec\alpha, n))}{n-1}  + \frac{1 + \sqrt{2 \pi}}{ n-1}   \Gamma(p+1),
\end{align*}
where $\Gamma(z) = \int_0^\infty x^{z-1} \exp( - x) \, dx$ is the gamma function. Combining this inequality with \eqref{eq:sin_alpha_bound} and using the basic upper bound $\Gamma(z+1) \le z^z$ completes the proof of Lemma~\ref{lem:alpha_bound}. 

\section{
\label{sec:proof:Delta2}
Proof of Lemma~\ref{lem:Delta2}}

\begin{lem}\label{lem:Delta2}
Let $A\sim \Ncal(0,1)$ and $B\sim \chi_{n-1}$. For all $p \ge 1$ and $n \geq 2$,
\begin{align*}
\left(\ex{ \left|\sqrt{A^2 + (B- \sqrt{n})^2}\right|^p}\right)^{1/p} \le C \sqrt{p},
\end{align*}
where $C\leq \sqrt{2} + e^{1/e}$.
\end{lem}

Define $\Delta_2 \coloneqq \sqrt{A^2 + (B-\sqrt{n})^2}$. Minkowski's inequality implies
\begin{align}
\left(\ex{ \Delta_2^p}\right)^{1/p} \le \ex{\Delta_2} + \left(\ex{ \left|\Delta_2- \ex{\Delta_2} \right|^p}\right)^{1/p}
\end{align}
Recall that $\mu_n$ of \eqref{eq:mu_n} is the mean of the chi-distribution with $n$ degrees of freedom. By Jensen's inequality, $\ex{\Delta_2}$ can be upper bounded as
\begin{align*}
\ex{\Delta_2}^2 & \le \ex{\Delta_2^2} \\
& = \var(A) + \var(B) + (\mu_{n-1}-\sqrt{n})^2 \\
& = 1 + (n-1)-\mu_{n-1}^2  + (\mu_{n-1}-\sqrt{n})^2 \leq 2,
\end{align*}
where we used that $\sqrt{n-1} \leq \mu_n < \sqrt{n}$. To bound the deviation about the mean, we use the fact that $B$ can be expressed as $B = \|W\|$ where $W$ is an $(n-1)$-dimensional vector with i.i.d.\ standard Gaussian entries. This allows us to write $\Delta_2 = \psi(A,W)$, where $\psi: \reals \times \reals^{n-1} \to \reals$ is given by $ \psi(a,w)  = (a^2 + \left(  \|w\| - \sqrt{n} \right)^2)^{1/2}$. The function $\psi$ is 1-Lipschitz continuous because it is the  composition of 1-Lipschitz functions. Therefore, we can apply the Tsirelson-Ibragimov-Sudakov Gaussian concentration inequality (see e.g., \cite[Theorem 5.2.2]{vershynin2018high}) to conclude that $\Delta_2 - \ex{\Delta_2}$ is sub-Gaussian with variance parameter one. Consequently
$\left(\ex{ |\Delta_2 - \ex{\Delta_2} |^p}\right)^{1/p} \le  e^{1/e} \sqrt{p}$. 
\QEDA

\color{black}

\section{Lipschitz Estimation in AWGN Model}

%This section provides the proof of the second part of Theorem~\ref{thm:indirect_source_coding}. The following result shows the conditional expectation in the AWGN model can be approximated by a Lipschitz estimator. 

\begin{lem}\label{lem:Lip_approx}
Suppose that $(U,X)$ are random vectors in $\reals^m \times \reals^n$ with finite second moments. Let $Z = X + \sigma W$ where $\sigma > 0$ is known and $W \sim \normal(0,I_n)$ is independent Gaussian noise. Define $f^*(z) = \ex{ U \mid Z  = z}$. For each $\eps > 0$ there exists a number $L$ and estimator $f: \reals^n \to \reals^m$ with $\|f\|_\mathrm{Lip}\le L$ such that
\begin{align}
\ex{ \|U - f(Z)\|^2} \le \ex{ \|U - f^*(Z)\|^2} + \eps. \label{eq:Lip_approx}
\end{align}
\end{lem}
\begin{proof}
Given $T > 0$, let $B$ be a binary random variable that is equal to zero if $ \{ \|U\| \vee \|X\| \le T\}$ and one otherwise. Define $g(z,b)= \ex{U \mid Z =z, B = b}$ to be the conditional expectation given $(Z,B)$ and let $f(z) = g(z,0)$. Two different applications of the law of total variance yields
\begin{align*}
\ex{ \|U - f^*(Z)\|^2}  &= \medmath{
\ex{ \|U - g(Z,B)\|^2}   + \ex{ \| g(Z,B) - f^*(Z)\|^2}}\\
\ex{ \|U - f(Z)\|^2}  & = \medmath{ \ex{ \|U - g(Z,B)\|^2}  + \ex{ \| g(Z,B) - f(Z)\|^2} } .
\end{align*}
Meanwhile, noting that $f(Z) = g(Z,B)$ whenever $B = 0$, we have
\begin{align*}
\|g(Z,B)- f(Z)\| & \le B  \| g(Z,B) - f(Z)\| \\
& \le B\| g(Z,B)\| +  B T 
\end{align*}
where the second step follows from the triangle inequality and that fact that $f(y)$ lies in a Euclidean ball of radius $T$. Combining the above displays with the inequality $(a+b)^2 \le 2 a^2 +2 b^2$ leads to
\begin{align*}
& \ex{ \|U - f(Z)\|^2}  -  \ex{ \|U - f^*(Z)\|^2} \\
& \qquad  \le  2 \ex{B \|g(Z, B)\|^2} + 2 \ex{ B} T^2\\
& \qquad  \le  2 \ex{\one(\|U\| > T) \| U \|^2} + 2 \pr{ \|U\| > T}  T^2. 
\end{align*}
By the assumption that $\|U\|$ has finite second moment, this upper bound converges to zero as $T$ increases. Thus, for each $\eps > 0$, there exists $T$ large enough such that \eqref{eq:Lip_approx} holds. 

Next, we will verify that $f$ has a finite Lipschitz constant. 
Lemma~\ref{lem:differential} below implies that the Jacobian of $f$ is given by
\begin{align*}
    \frac{\partial f(z)}{ \partial z} = \frac{ \cov( U, X \mid Z = z, B = 0)}{ \sigma^2}.
\end{align*}
By the Cauchy-Schwarz inequality and the definition of $B$, it follows that $\|\cov( U, X \mid Z = z, B = 0)\| \le T^2$, uniformly for all $z$, and thus $\|f\|_\mathrm{Lip} \le T^2/\sigma^2.$
\end{proof}

\begin{lem}\label{lem:differential}
Let $X$ be a $n$-dimensional random vector with $\ex{ \rho\left(X \right)} < \infty$, where $\rho(x)$ is the standard Gaussian density in $n$ dimensions, and let $Y\sim \Ncal(X,\sigma^2I_n)$. Let $h : \mathbb R^n \to \mathbb R^m$ be a measurable function such that 
\[
\phi(y) \coloneqq \mathbb E \left[ h(X) | Y = y \right],\quad y \in \mathbb R^n,
\]
is defined for any $y \in \reals^n$. The Jacobian of $\phi$  is given by
\[
J_{\phi}(y) = \frac{1}{\sigma^2}\cov \left( X,h(X) \mid Y = y\right).
\]
\end{lem}

\subsubsection*{Proof of Lemma~\ref{lem:differential}}
Set $\rho_\sigma(x) \coloneqq \rho(x/\sigma)/\sigma$, 
$P_Y(y) \coloneqq \ex{ \rho_{\sigma}\left(X-y \right)}$, and 
$x_y \coloneqq \ex{X | Y=y }$. 
From Bayes rule we have
\[
\phi(y) = \frac{\ex{ h(X)\rho_{\sigma}\left(X-y \right)}} {
\ex{ \rho_{\sigma}\left(X-y \right) }} = \frac{ \ex{  h(X)\rho_{\sigma}\left(X-y \right)}} {
P_{Y}(y)}. 
\]
It follows from \cite[Thm. 2.7.1]{lehmann2006testing} that we may differentiate with respect to $y_j$ within the expectation. We get
\begin{align*}
& [J_{\phi}(y)]_{i,j} = \frac{\partial \phi_i}{\partial y_j}  \\
& =  \frac{ \ex{  h_i(X) \rho_{\sigma} \left( X-y  \right) (X_j-y_j) } }
{ \sigma^2 P_{Y}(y)}
\\ 
& \quad - \frac{ \ex{ h_i(X) \rho_{\sigma} \left( X-y  \right) } \ex{ \rho_{\sigma} \left(X-y \right) (X_j-y_j) }} {\sigma^2 P_{Y}^2(y) } \\
& = \frac{ \ex{ \rho_{\sigma} \left(X-y \right) \left( h_i(X) - \phi_i(y) \right) (X_j-y_j) } } {
\sigma^2 P_{Y}(y)}  \\
& =  \frac{ \ex{ \rho_{\sigma}^2 \left(X-y \right) \left( h_i(X) - \phi_i(y) \right) X_j } } {
\sigma^2 P_{Y}(y)}. 
\end{align*}
The last transition implies  
\begin{align*}
& \frac{\partial \phi_i}{\partial y_j} = 
 \ex {\frac{\rho_{\sigma} \left(X-y \right)}{\sigma^2 P_Y(y)} \left( h_i(X) - \phi_i(y) \right) (X_j-a)},
\end{align*}
for any constant $a\in \mathbb R$. It follows that
\begin{align*}
J_{\phi}(y) & =\ex{ \left(X - x_y \right) \left(h(X) - \phi(y)\right)^T \frac{\rho_{\sigma}(X-y)}{ \sigma^2 P_Y(y)} } \\
& \frac{1}{\sigma^2} \int_{\mathbb R^n} \left(x-{x_y} \right) \left(h(x)-\phi(y)\right)^T P_{X|Y}(dx|Y=y) \\
& = \frac{1}{\sigma^2} \mathbb E \left[ \left(X-x_y \right) \left(h(X)-\phi(y)\right)^T  \mid Y=y\right] \\
& = \frac{1}{\sigma^2} \cov \left(X,h(X) | Y=y \right).
\end{align*}

\section{Lipschtiz Continuity of AMP \label{app:AMP}}

\begin{prop}
Let $A \in \reals^{n \times d_n}$ be a random matrix with i.i.d. entries $\Ncal(0,1/n)$. Assume that $n/d_n\to \delta \in (0,\infty)$. Denote by $z\to \theta^{t}_{\AMP}(z)$ the result of $t$ iterations of AMP using a sequence of local non-linearity functions $\{\eta_1,\ldots,\eta_t\}$ as in \eqref{eq:SE} and set
$L^{\AMP}_{n,t} = \|\theta^t_{\AMP}\|_{\Lip}$. 
If $\eta_k$ is $L_k$-Lipschitz for $k=1,\ldots,t$, then with probability one there exists $K_t$ such that $\sup_n L^{\AMP}_{n,t} \leq K_t$.
\end{prop}

\begin{proof}
Use the tail bound on the maximal eigenvalue of a random matrix with sub-Gaussian entries from \cite[Thm 4.4.5]{vershynin2018high} to deduce that there exists a constant $c$, independent of $n$, such that 
\[
\Pr\left( \sqrt{n}\|A\|_2 \leq c(\sqrt{n}(1+\sqrt{\delta})+a)  \right) \geq 1-2e^{-a^2}.
\]
For $C = c(2+\sqrt{\delta})$, define the event
\[
E_n = \{ \|A\|_2 \leq C \}.
\]
Using $a=\sqrt{n}$, the Borel-Cantelli Lemma applied to the sequence $\{E_n\}$ implies that the event
\[
G \coloneqq \left\{ \exists n_0 \,:\,\|A\|_2 \leq C,\quad \forall n\geq n_0 \right\}
\]
occurs with probability one. Conditioning on $G$ and given such $n_0$, we consider the $t$-th iteration of AMP for reconstructing $\theta$ from $z = A\theta + W$ as given by \eqref{eq:AMP1} and \eqref{eq:AMP2}. For $\eta_t$ applied element-wise to vectors $x,r \in \reals^n$, we have
\[
\left\| \eta_t(u) - \eta_t(x) \right\| \leq L_t \left\|u-x\right\|,
\]
and 
\[
\left| \langle 
     \eta'_t (x) \rangle \right| \leq L_t. 
\]
For $x \in \reals^n$, denote by $\theta^t(x)$ and $r^t(x)$
the result of applying $t$ iterations of \eqref{eq:AMP2} to $x$. We have
\begin{align}
    & \left\| \theta^{t+1}(x) - \theta^{t+1}(\tilde{x})  \right\| \nonumber \\
    & = \left\| \eta_t \left( A^\top r^t(x) + \theta^t(x) \right) - \eta_t \left( A^\top r^t(\tilde{x}) + \theta^t(\tilde{x}) \right) \right\| \nonumber \\
    & \leq L_t  \left\|  A^\top \left( r^t(x) -  r^t(\tilde{x}) \right) + \theta^t(x) - \theta^t(\tilde{x}) \right\| \nonumber \\
    & \leq L_t \left\|  A^\top \left( r^t(x) -  r^t(\tilde{x}) \right)  \right\| + L_t \left\| \theta^t(x) - \theta^t(\tilde{x}) \right\| \nonumber \\
    & \leq L_t C \left\| r^t(x) -  r^t(\tilde{x})  \right\| + L_t \left\| \theta^t(x) - \theta^t(\tilde{x}) \right\|. \label{eq:th_rec}
\end{align}
Furthermore,
\begin{align}
    & \left\| r^t(x) - r^t(\tilde{x})  \right\| \leq  \left\| x- \tilde{x} \right\| +  \left\| A \left(\theta^t(x) - \theta^t(\tilde{x}) \right) \right\| \nonumber \\
    & \quad + \frac{1}{n}\left\| r^{t-1}(x) \sum_{i=1}^n \eta'_{t-1} \left([A^\top r^{t-1}(x)+\theta^t(x)]_i  \right) \right. \nonumber \\
    & \qquad \qquad  - \left.
    r^{t-1}(\tilde{x}) \sum_{i=1}^n \eta'_{t-1} \left([A^\top r^{t-1}(\tilde{x})+\theta^t(\tilde{x})]_i  \right) \right\| \nonumber \\
    & \leq \left\| x- \tilde{x} \right\| + C \left\| \theta^t(x) - \theta^t(\tilde{x}) \right\|\nonumber \\
    & \qquad \qquad + L_{t-1} \left\|r^{t-1}(x) - r^{t-1}(\tilde{x}) \right\| \label{eq:u_rec} .
\end{align}
We now prove by induction that for $t=1,\ldots$, there exists $K_t$ and $R_t$ such that
\begin{align}
    &\left\| \theta^t(x) - \theta^t(\tilde{x}) \right\| \leq K_t \left\|x - \tilde{x} \right\| \label{eq:x_t}\\
    & \left\| r^{t-1}(x) - r^{t-1}(\tilde{x}) \right\| \leq R_{t-1} \left\|x - \tilde{x} \right\|. \label{eq:u_t}
    \end{align}
For $t=1$, we have
\begin{align*}
    & \left\| \theta^0(x) - \theta^0(\tilde{x}) \right\| = 0, \\
    &\left\| u^0(x) - u^0(\tilde{x}) \right\| = \left\| x - \tilde{x} \right\|, \\
    &\left\| \theta^1(x) - \theta^1(\tilde{x}) \right\| \leq L_0 C \left\|x - \tilde{x} \right\|,
\end{align*}
and for the second inequality we take $R_0 =0$. Assume now that for all $k=1,\ldots,t-1$, there are $K_k$ and $R_k$ such that
\begin{align*}
    &\left\| \theta^k(x) - \theta^k(\tilde{x}) \right\| \leq K_k \left\|x - \tilde{x} \right\|, \\
    & \left\| r^{k-1}(x) - r^{k-1}(\tilde{x}) \right\| \leq R_{k-1} \left\|x - \tilde{x} \right\|. 
    \end{align*}
From \eqref{eq:th_rec}, \eqref{eq:x_t}, and \eqref{eq:u_t}, we obtain
\begin{align*}
& \left\| \theta^t(x) - \theta^t(\tilde{x}) \right\| \leq L_{t-1} C \left\| r^{t-1}(x) - r^{t-1}(\tilde{x}) \right\| \\
& \qquad + L_{t-1} \left\| \theta^{t-1}(x) - \theta^{t-1}(\tilde{x}) \right\|, \\
& \leq L_{t-1} C R_{t-1} \left\|x-\tilde{x}\right\| + L_{t-1} K_{t-1} \left\|x-\tilde{x}\right\| \\
& = \left( L_{t-1}C R_{t-1}+L_{t-1} K_{t-1} \right)  \left\|x-\tilde{x}\right\|.
\end{align*}
From \eqref{eq:u_rec}, \eqref{eq:x_t}, and \eqref{eq:u_t}, we obtain
\begin{align*}
& \left\| r^{t-1}(x) - r^{t-1}(\tilde{x}) \right\| \leq \left\|x-\tilde{x} \right\| + C \left\| \theta^{t-1}(x) - \theta^{t-1}(\tilde{X}) \right\| \\
& \qquad + L_{t-2} \left\|r^{t-2}(x)-r^{t-2}(\tilde{x}) \right\| \\
& \leq \left( 1 +  C_A K_{t-1} + L_{t-2} R_{t-2} \right)\left\|x-\tilde{x} \right\|.
\end{align*}
It follows that both \eqref{eq:x_t} and \eqref{eq:u_t} hold with $k=t$. \par
We have shown that with probability one there exists $n_0$ such that, for each $t \in \mathbb N$, there exists $K_t$ for which
\begin{align}
    &\left\| \theta^t(x) - \theta^t(\tilde{x}) \right\| \leq K_t \left\|x - \tilde{x} \right\|,\quad \forall n\geq n_0.
\end{align}
It follows that for each $t$, 
$\sup_n L^{\AMP}_{n,t}\leq K_t$.
\end{proof}

\bibliographystyle{IEEEtran}
\bibliography{IEEEfull,SphericalCodes}

\begin{thebibliography}{10}
\providecommand{\url}[1]{#1}
\csname url@samestyle\endcsname
\providecommand{\newblock}{\relax}
\providecommand{\bibinfo}[2]{#2}
\providecommand{\BIBentrySTDinterwordspacing}{\spaceskip=0pt\relax}
\providecommand{\BIBentryALTinterwordstretchfactor}{4}
\providecommand{\BIBentryALTinterwordspacing}{\spaceskip=\fontdimen2\font plus
\BIBentryALTinterwordstretchfactor\fontdimen3\font minus
  \fontdimen4\font\relax}
\providecommand{\BIBforeignlanguage}[2]{{%
\expandafter\ifx\csname l@#1\endcsname\relax
\typeout{** WARNING: IEEEtran.bst: No hyphenation pattern has been}%
\typeout{** loaded for the language `#1'. Using the pattern for}%
\typeout{** the default language instead.}%
\else
\language=\csname l@#1\endcsname
\fi
#2}}
\providecommand{\BIBdecl}{\relax}
\BIBdecl

\bibitem{KipnisReevesISIT2019}
A.~{Kipnis} and G.~{Reeves}, ``Gaussian approximation of quantization error for
  estimation from compressed data,'' in \emph{2019 IEEE International Symposium
  on Information Theory (ISIT)}, July 2019, pp. 2029--2033.

\bibitem{gray1998quantization}
R.~Gray and D.~Neuhoff, ``Quantization,'' \emph{{IEEE} Transactions on
  Information Theory}, vol.~44, no.~6, pp. 2325--2383, Oct 1998.

\bibitem{BLTJ1340}
W.~R. Bennett, ``Spectra of quantized signals,'' \emph{Bell System Technical
  Journal}, vol.~27, no.~3, pp. 446--472, 1948.

\bibitem{1424312}
D.~Marco and D.~Neuhoff, ``The validity of the additive noise model for uniform
  scalar quantizers,'' \emph{{IEEE} Transactions on Information Theory},
  vol.~51, no.~5, pp. 1739--1755, May 2005.

\bibitem{370112}
R.~Zamir and M.~Feder, ``Rate-distortion performance in coding bandlimited
  sources by sampling and dithered quantization,'' \emph{{IEEE} Transactions on
  Information Theory}, vol.~41, no.~1, pp. 141--154, Jan 1995.

\bibitem{lee1996asymptotic}
D.~H. Lee and D.~L. Neuhoff, ``Asymptotic distribution of the errors in scalar
  and vector quantizers,'' \emph{{IEEE} Transactions on Information Theory},
  vol.~42, no.~2, pp. 446--460, 1996.

\bibitem{kontoyiannis2006mismatched}
I.~Kontoyiannis and R.~Zamir, ``Mismatched codebooks and the role of entropy
  coding in lossy data compression,'' \emph{{IEEE} Transactions on Information
  Theory}, vol.~52, no.~5, pp. 1922--1938, 2006.

\bibitem{tsitsiklis1988decentralized}
J.~N. Tsitsiklis, ``Decentralized detection by a large number of sensors,''
  \emph{Mathematics of Control, Signals, and Systems (MCSS)}, vol.~1, no.~2,
  pp. 167--182, 1988.

\bibitem{zhang1988estimation}
Z.~Zhang and T.~Berger, ``Estimation via compressed information,'' \emph{{IEEE}
  Transactions on Information Theory}, vol.~34, no.~2, pp. 198--211, 1988.

\bibitem{HanAmari1998}
T.~Han and S.~Amari, ``Statistical inference under multiterminal data
  compression,'' \emph{{IEEE} Transactions on Information Theory}, vol.~44,
  no.~6, pp. 2300--2324, Oct 1998.

\bibitem{steinhardt2015minimax}
J.~Steinhardt and J.~Duchi, ``Minimax rates for memory-bounded sparse linear
  regression,'' in \emph{Conference on Learning Theory}, 2015, pp. 1564--1587.

\bibitem{zhu2017quantized}
Y.~Zhu and J.~Lafferty, ``Quantized minimax estimation over sobolev
  ellipsoids,'' \emph{Information and Inference: A Journal of the IMA}, vol.~7,
  no.~1, pp. 31--82, 2017.

\bibitem{KipnisDuchi}
\BIBentryALTinterwordspacing
A.~Kipnis and J.~C. Duchi, ``Mean estimation from one-bit measurements,''
  \emph{CoRR}, vol. abs/1901.03403, 2019. [Online]. Available:
  \url{http://arxiv.org/abs/1901.03403}
\BIBentrySTDinterwordspacing

\bibitem{han2018distributed}
Y.~Han, P.~Mukherjee, A.~Ozgur, and T.~Weissman, ``Distributed statistical
  estimation of high-dimensional and nonparametric distributions,'' in
  \emph{2018 IEEE International Symposium on Information Theory (ISIT)}.\hskip
  1em plus 0.5em minus 0.4em\relax IEEE, 2018, pp. 506--510.

\bibitem{dagan2018detecting}
Y.~Dagan and O.~Shamir, ``Detecting correlations with little memory and
  communication,'' in \emph{Conference On Learning Theory}.\hskip 1em plus
  0.5em minus 0.4em\relax PMLR, 2018, pp. 1145--1198.

\bibitem{szabo2018adaptive}
B.~Szabo and H.~van Zanten, ``Adaptive distributed methods under communication
  constraints,'' \emph{arXiv preprint arXiv:1804.00864}, 2018.

\bibitem{barnes2019learning}
L.~P. Barnes, Y.~Han, and A.~Ozgur, ``Lower bounds for learning distributions
  under communication constraints via fisher information,'' \emph{Journal of
  Machine Learning Research}, vol.~21, no. 236, pp. 1--30, 2020.

\bibitem{DobrushinTsybakov}
R.~Dobrushin and B.~Tsybakov, ``Information transmission with additional
  noise,'' \emph{{IRE} Transactions on Information Theory}, vol.~8, no.~5, pp.
  293--304, 1962.

\bibitem{berger1971rate}
T.~Berger, \emph{Rate-distortion theory: A mathematical basis for data
  compression}.\hskip 1em plus 0.5em minus 0.4em\relax Englewood Cliffs, NJ:
  Prentice-Hall, 1971.

\bibitem{1054469}
J.~Wolf and J.~Ziv, ``Transmission of noisy information to a noisy receiver
  with minimum distortion,'' \emph{{IEEE} Transactions on Information Theory},
  vol.~16, no.~4, pp. 406--411, 1970.

\bibitem{1056251}
H.~Witsenhausen, ``Indirect rate distortion problems,'' \emph{{IEEE}
  Transactions on Information Theory}, vol.~26, no.~5, pp. 518--521, 1980.

\bibitem{dembo2003minimax}
A.~Dembo and T.~Weissman, ``The minimax distortion redundancy in noisy source
  coding,'' \emph{Information Theory, IEEE Transactions on}, vol.~49, no.~11,
  pp. 3020--3030, 2003.

\bibitem{weissman2004universally}
T.~Weissman, ``Universally attainable error exponents for rate-distortion
  coding of noisy sources,'' \emph{IEEE Transactions on Information Theory},
  vol.~50, no.~6, pp. 1229--1246, 2004.

\bibitem{zhu2014quantized}
Y.~Zhu and J.~Lafferty, ``Quantized estimation of gaussian sequence models in
  euclidean balls,'' in \emph{Advances in Neural Information Processing
  Systems}, 2014, pp. 3662--3670.

\bibitem{Donoho1994}
D.~L. Donoho and I.~M. Johnstone, ``Minimax risk over lp-balls for lp-error,''
  \emph{Probability Theory and Related Fields}, vol.~99, no.~2, pp. 277--303,
  Jun 1994.

\bibitem{johnstone2011gaussian}
I.~Johnstone, ``Gaussian estimation: sequence and multiresolution models,''
  \emph{Unpublished manuscript}, 2011.

\bibitem{eldar2012compressed}
Y.~C. Eldar and G.~Kutyniok, \emph{Compressed sensing: theory and
  applications}.\hskip 1em plus 0.5em minus 0.4em\relax Cambridge University
  Press, 2012.

\bibitem{donoho2009message}
D.~L. Donoho, A.~Maleki, and A.~Montanari, ``Message-passing algorithms for
  compressed sensing,'' \emph{Proceedings of the National Academy of Sciences},
  vol. 106, no.~45, pp. 18\,914--18\,919, 2009.

\bibitem{goyal2008compressive}
V.~K. Goyal, A.~K. Fletcher, and S.~Rangan, ``Compressive sampling and lossy
  compression,'' \emph{{IEEE} Signal Processing Magazine}, vol.~25, no.~2, pp.
  48--56, 2008.

\bibitem{baraniuk2017exponential}
R.~G. Baraniuk, S.~Foucart, D.~Needell, Y.~Plan, and M.~Wootters, ``Exponential
  decay of reconstruction error from binary measurements of sparse signals,''
  \emph{{IEEE} Transactions on Information Theory}, vol.~63, no.~6, pp.
  3368--3385, 2017.

\bibitem{kipnis2017fundamental}
A.~Kipnis, G.~Reeves, Y.~C. Eldar, and A.~J. Goldsmith, ``Fundamental limits of
  compressed sensing under optimal quantization,'' in \emph{Information Theory
  (ISIT), 2017 IEEE International Symposium on}, 2017.

\bibitem{8356140}
M.~{Leinonen}, M.~{Codreanu}, M.~{Juntti}, and G.~{Kramer}, ``Rate-distortion
  performance of lossy compressed sensing of sparse sources,'' \emph{IEEE
  Transactions on Communications}, vol.~66, no.~10, pp. 4498--4512, Oct 2018.

\bibitem{sakrison1968geometric}
D.~Sakrison, ``A geometric treatment of the source encoding of a {G}aussian
  random variable,'' \emph{IEEE Transactions on Information Theory}, vol.~14,
  no.~3, pp. 481--486, 1968.

\bibitem{wyner:1968}
A.~D. Wyner, ``Communication of analog data from a {G}aussian source over a
  noisy channel,'' \emph{The Bell System Technical Journal}, vol.~47, no.~5,
  pp. 801--812, 1968.

\bibitem{lapidoth1997role}
A.~Lapidoth, ``On the role of mismatch in rate distortion theory,''
  \emph{{IEEE} Transactions on Information Theory}, vol.~43, no.~1, pp. 38--47,
  1997.

\bibitem{zhou2017refined}
L.~{Zhou}, V.~Y.~F. {Tan}, and M.~{Motani}, ``Refined asymptotics for
  rate-distortion using {G}aussian codebooks for arbitrary sources,''
  \emph{IEEE Transactions on Information Theory}, vol.~65, no.~5, pp.
  3145--3159, May 2019.

\bibitem{delsarte1991spherical}
P.~Delsarte, J.-M. Goethals, and J.~J. Seidel, ``Spherical codes and designs,''
  in \emph{Geometry and Combinatorics}.\hskip 1em plus 0.5em minus 0.4em\relax
  Elsevier, 1991, pp. 68--93.

\bibitem{ThomasCover}
T.~M. Cover and J.~A. Thomas, \emph{Elements of information theory (2.
  ed.)}.\hskip 1em plus 0.5em minus 0.4em\relax Wiley, 2006.

\bibitem{han1987hypothesis}
T.~S. Han, ``Hypothesis testing with multiterminal data compression,''
  \emph{{IEEE} Transactions on Information Theory}, vol.~33, no.~6, pp.
  759--772, 1987.

\bibitem{duchi2014optimality}
Y.~Zhang, J.~Duchi, M.~I. Jordan, and M.~J. Wainwright, ``Information-theoretic
  lower bounds for distributed statistical estimation with communication
  constraints,'' in \emph{Advances in Neural Information Processing Systems},
  2013, pp. 2328--2336.

\bibitem{kipnis2021rate}
A.~Kipnis, S.~Rini, and A.~J. Goldsmith, ``The rate-distortion risk in
  estimation from compressed data,'' \emph{IEEE Transactions on Information
  Theory}, vol.~67, no.~5, pp. 2910--2924, 2021.

\bibitem{KipnisWiener2019}
A.~Kipnis, A.~J. Goldsmith, and Y.~C. Eldar, ``The distortion-rate function of
  sampled {W}iener processes,'' \emph{IEEE Transactions on Information Theory},
  vol.~65, no.~1, pp. 482--499, Jan 2019.

\bibitem{rachev1998mass}
S.~T. Rachev and L.~R{\"u}schendorf, \emph{Mass Transportation Problems: Volume
  I: Theory}.\hskip 1em plus 0.5em minus 0.4em\relax Springer Science \&
  Business Media, 1998, vol.~1.

\bibitem{ambrosio2003optimal}
L.~Ambrosio, Y.~Brenier, G.~Buttazzo, and C.~Villani, \emph{Optimal
  Transportation and Applications: Lectures given at the CIME Summer School
  held in Martina Franca, Italy, September 2--8, 2001}.\hskip 1em plus 0.5em
  minus 0.4em\relax Springer, 2003.

\bibitem{villani2008optimal}
C.~Villani, \emph{Optimal transport: old and new}.\hskip 1em plus 0.5em minus
  0.4em\relax Springer Science \& Business Media, 2008, vol. 338.

\bibitem{pollard1982quantization}
D.~Pollard, ``Quantization and the method of k-means,'' \emph{IEEE Transactions
  on Information theory}, vol.~28, no.~2, pp. 199--205, 1982.

\bibitem{linder2002lagrangian}
T.~Linder, ``Lagrangian empirical design of variable-rate vector quantizers:
  consistency and convergence rates,'' \emph{IEEE Transactions on Information
  Theory}, vol.~48, no.~11, pp. 2998--3003, 2002.

\bibitem{1056045}
R.~{Gray} and D.~{Ornstein}, ``Block coding for discrete stationary
  d-continuous noisy channels,'' \emph{IEEE Transactions on Information
  Theory}, vol.~25, no.~3, pp. 292--306, May 1979.

\bibitem{gray1980block}
R.~Gray, D.~Ornstein, and R.~Dobrushin, ``Block synchronization, sliding-block
  coding, invulnerable sources and zero error codes for discrete noisy
  channels,'' \emph{The Annals of Probability}, pp. 639--674, 1980.

\bibitem{ZamirFeder1996}
R.~Zamir and M.~Feder, ``On lattice quantization noise,'' \emph{{IEEE}
  Transactions on Information Theory}, vol.~42, no.~4, pp. 1152--1159, Jul
  1996.

\bibitem{RaginskySason2014}
M.~Raginsky, I.~Sason \emph{et~al.}, ``Concentration of measure inequalities in
  information theory, communications, and coding,'' \emph{Foundations and
  Trends{\textregistered} in Communications and Information Theory}, vol.~10,
  no. 1-2, pp. 1--246, 2013.

\bibitem{harsha2007communication}
P.~Harsha, R.~Jain, D.~McAllester, and J.~Radhakrishnan, ``The communication
  complexity of correlation,'' in \emph{Twenty-Second Annual IEEE Conference on
  Computational Complexity (CCC'07)}.\hskip 1em plus 0.5em minus 0.4em\relax
  IEEE, 2007, pp. 10--23.

\bibitem{cuff2013distributed}
P.~Cuff, ``Distributed channel synthesis,'' \emph{IEEE Transactions on
  Information Theory}, vol.~59, no.~11, pp. 7071--7096, 2013.

\bibitem{LiElGamal2018}
C.~T. {Li} and A.~{El Gamal}, ``A universal coding scheme for remote generation
  of continuous random variables,'' \emph{IEEE Transactions on Information
  Theory}, vol.~64, no.~4, pp. 2583--2592, April 2018.

\bibitem{stam:1982}
A.~J. Stam, ``Limit theorems for uniform distributions on spheres in
  high-dimensional euclidean spaces,'' \emph{Journal of Applied Probability},
  vol.~19, no.~1, pp. 221--228, Mar. 1982.

\bibitem{vershynin2018high}
R.~Vershynin, \emph{High-dimensional probability: An introduction with
  applications in data science}.\hskip 1em plus 0.5em minus 0.4em\relax
  Cambridge University Press, 2018, vol.~47.

\bibitem{BayatiMontanari2011}
M.~Bayati and A.~Montanari, ``The dynamics of message passing on dense graphs,
  with applications to compressed sensing,'' \emph{{IEEE} Transactions on
  Information Theory}, vol.~57, no.~2, pp. 764--785, Feb 2011.

\bibitem{gray2011entropy}
R.~M. Gray, \emph{Entropy and information theory}.\hskip 1em plus 0.5em minus
  0.4em\relax Springer-Verlag, 2011, vol.~1.

\bibitem{7464359}
V.~{Kostina} and S.~{Verdú}, ``Nonasymptotic noisy lossy source coding,''
  \emph{IEEE Transactions on Information Theory}, vol.~62, no.~11, pp.
  6111--6123, Nov 2016.

\bibitem{kipnis2017analog}
A.~Kipnis, ``Fundamental distortion limits of analog-to-digital compression,''
  Ph.D. dissertation, Stanford University, 2017.

\bibitem{Blahut}
R.~Blahut, ``Computation of channel capacity and rate-distortion functions,''
  \emph{{IEEE} Transactions on Information Theory}, vol.~18, no.~4, pp.
  460--473, Jul 1972.

\bibitem{lehmann2006testing}
E.~L. Lehmann and J.~P. Romano, \emph{Testing statistical hypotheses}.\hskip
  1em plus 0.5em minus 0.4em\relax Springer Science \& Business Media, 2006.

\end{thebibliography}

\end{document}